\documentclass[a4paper,USenglish]{lipics-v2021}
\usepackage[utf8]{inputenc}

\long\def\commentbegin #1\commentend{}

\usepackage{balance}
\usepackage{algorithm}
\usepackage[noend]{algpseudocode}
\usepackage{ifthen}
\usepackage{comment}
\usepackage{amsthm,amsmath}
\usepackage{thmtools,thm-restate}
\usepackage{graphicx}
\usepackage{bbm}

\usepackage{graphicx}%
\usepackage{multirow}%
\usepackage{amsmath,amssymb,amsfonts}%
\usepackage{amsthm}%
\usepackage{mathrsfs}%
\usepackage[title]{appendix}%
\usepackage{xcolor}%
\usepackage{textcomp}%
\usepackage{manyfoot}%
\usepackage{booktabs}%
\usepackage{algorithm}%
\usepackage{algorithmicx}%
\usepackage{algpseudocode}%
\usepackage{listings}%

\usepackage{pstricks,pst-node,pst-tree}
\usepackage{soul}
\usepackage{subcaption}
\usepackage{forest}

\usepackage{color}
\usepackage{mathrsfs}
\usepackage[normalem]{ulem}
\usepackage{paralist}

\renewcommand{\epsilon}{\varepsilon}
\newcommand{\Prob}[1]{\hbox{\rm I\kern-2pt P}\left[#1\right]}
\def\poly{\operatorname{poly}}

\DeclareMathAlphabet{\mathsc}{OT1}{cmr}{m}{sc}

\renewcommand{\geq}{\geqslant}

\renewcommand{\leq}{\leqslant}

\newboolean{short}
\setboolean{short}{false}

\algblockdefx[ExecBl]{BlockOn}{BlockOff}  [1]{#1}
  
\makeatletter
\ifthenelse{\equal{\ALG@noend}{t}}%
  {\algtext*{BlockOff}}
  {}%
\makeatother

\algtext*{EndWhile} 
\algtext*{EndIf} 
\algtext*{EndFor} 

\newcommand{\shortOnly}[1]{\ifthenelse{\boolean{short}}{#1}{}}
\newcommand{\onlyShort}[1]{\ifthenelse{\boolean{short}}{}{#1}}
\newcommand{\longOnly}[1]{\ifthenelse{\boolean{short}}{}{#1}}
\newcommand{\onlyLong}[1]{\ifthenelse{\boolean{short}}{}{#1}}

\def\ShowComment{True}
\ifdefined\ShowComment
\def\billy#1{{\color{green}\underline{\textsf{Billy:}}} {\color{blue} \emph{#1}}}
\def\gopal#1{{\color{red}\underline{\textsf{Gopal:}}} {\color{blue} \emph{#1}}}
\def\fabien#1{{\color{orange}\underline{\textsf{Fabien:}}} {\color{blue} \emph{#1}}}
\else
\def\billy#1{}
\def\gopal#1{}
\def\fabien#1{}
\fi

\def\inline#1:{\par\vskip 7pt\noindent{\bf #1:}\hskip 10pt}

\def\CONGEST{\textsc{CONGEST}}
\def\LOCAL{\textsc{LOCAL}}

\def\TS{Transmission-Schedule}

\def\UPCASTMIN{\textsc{Upcast-Min}}

\def\DOWNCAST{\textsc{Fragment-Broadcast}}
\def\TRANSMITADJACENT{\textsc{Transmit-Adjacent}}

\def\LDTBROADCAST{\textsc{LDT-Broadcast}}

\def\LDTCONSTRUCTAWAKE{\textsc{LDT-Construct-Awake}}
\def\LDTCONSTRUCTROUND{\textsc{LDT-Construct-Round}}
\def\AuxMISone{\textsc{VT-MIS}}
\def\AuxMIStwo{\textsc{LDT-MIS}}
\def\AuxMIStwoROUND{\textsc{LDT-MIS-ROUND}}
\def\AwakeMIS{\textsc{Awake-MIS}}

\def\LDTRANK{\textsc{LDT-Ranking}}

\def\DOWNRECEIVE{\textsf{Down-Receive}}
\def\DOWNSEND{\textsf{Down-Send}}
\def\SIDESENDRECEIVE{\textsf{Side-Send-Receive}}
\def\UPRECEIVE{\textsf{Up-Receive}}
\def\UPSEND{\textsf{Up-Send}}

\renewcommand{\geq}{\geqslant}

\renewcommand{\leq}{\leqslant}

\long\def\hide #1\hideend{}

\newcommand{\squishlist}{
 \begin{list}{$\bullet$}
  { \setlength{\itemsep}{0pt}
     \setlength{\parsep}{3pt}
     \setlength{\topsep}{3pt}
     \setlength{\partopsep}{0pt}
     \setlength{\leftmargin}{1.5em}
     \setlength{\labelwidth}{1em}
     \setlength{\labelsep}{0.5em} } }

\newcommand{\squishlisttwo}{
 \begin{list}{$\bullet$}
  { \setlength{\itemsep}{0pt}
     \setlength{\parsep}{0pt}
    \setlength{\topsep}{0pt}
    \setlength{\partopsep}{0pt}
    \setlength{\leftmargin}{2em}
    \setlength{\labelwidth}{1.5em}
    \setlength{\labelsep}{0.5em} } }

\newcommand{\squishend}{
  \end{list}  }

\title{Distributed MIS in \texorpdfstring{$O(\log\log{n} )$}{O(log log n)} Awake Complexity}

\author{Fabien Dufoulon\footnote{Part of the work was done while Fabien Dufoulon was a postdoctoral fellow at the University of Houston in Houston, USA.
}}{School of Computing and Communications, Lancaster University, Lancaster, UK}{fabien.dufoulon.cs@gmail.com}{https://orcid.org/0000-0003-2977-4109}{Supported in part by National Science Foundation (NSF) grants CCF-1540512, IIS-1633720, and CCF-1717075 and U.S.-Israel Binational Science Foundation (BSF) grant 2016419.}

\author{William K. Moses Jr.\footnote{Part of the work was done while William K. Moses Jr.\ was a postdoctoral fellow at the University of Houston in Houston, USA.
}}{Department of Computer Science, Durham University, Durham, UK}{wkmjr3@gmail.com}{https://orcid.org/0000-0002-4533-7593}{Supported in part by NSF grants CCF-1540512, IIS-1633720, and CCF-1717075 and BSF grant 2016419.}

\author{Gopal Pandurangan}{Department of Computer Science, University of Houston, Houston,  TX, USA}{gopalpandurangan@gmail.com}{https://orcid.org/0000-0001-5833-6592}{Supported in part by NSF grants CCF-1540512, IIS-1633720, and CCF-1717075 and BSF grant 2016419.}

\authorrunning{F. Dufoulon, W.\,K. Moses Jr., and G. Pandurangan} 

\Copyright{Fabien Dufoulon, William K. Moses Jr., and Gopal Pandurangan}



\nolinenumbers
\hideLIPIcs

\begin{document}

\maketitle

\begin{abstract}
Maximal Independent Set (MIS) is one of the fundamental and most well-studied problems  in distributed graph algorithms. 
Even after four decades of intensive research,  the best known (randomized) MIS algorithms have $O(\log{n})$ round complexity on general graphs [Luby, STOC 1986] (where $n$ is the number of nodes), while  the best known lower bound is $\Omega(\sqrt{\log{n}/\log\log{n}})$ [Kuhn, Moscibroda,  Wattenhofer, JACM 2016]. Breaking past the $O(\log{n})$ round complexity upper bound  or showing stronger lower bounds  have been 
longstanding open problems. 

Energy is a premium resource in various settings such as  battery-powered wireless networks and sensor networks. The bulk of the energy is used by nodes when they are \emph{awake}, i.e., when they are sending, receiving, and even just listening for messages. On the other hand, when a node is \emph{sleeping}, it does not perform any  communication and thus spends very little energy.  Several recent works  have addressed the problem of designing \emph{energy-efficient} distributed algorithms for various fundamental problems. These algorithms operate by  minimizing
the   number of rounds in which {\em any} node is  \emph{awake},  also called the (worst-case) \emph{awake complexity}.  
An intriguing open question is whether one can design
 a distributed MIS algorithm that has significantly  smaller awake complexity compared to existing algorithms. In particular, the question of obtaining a distributed MIS algorithm with $o(\log n)$ 
 awake complexity was  left open in [Chatterjee, Gmyr, Pandurangan, PODC 2020].

Our main contribution is to show that MIS can be computed in  awake complexity that is  \emph{exponentially} better compared to the best known  round complexity of $O(\log n)$ 
 and also bypassing its fundamental $\Omega(\sqrt{\log{n}/\log\log{n}})$  round complexity lower bound  exponentially. 
Specifically, we show that MIS can be computed by a randomized distributed (Monte Carlo) algorithm  in $O(\log\log{n} )$ awake  complexity  with high probability (i.e., with probability at least $1 - n^{-1}$). This algorithm has a round complexity of $O((\log^7 n) \log \log n)$. 
We also show that we can improve the  round complexity at the cost of a slight increase in awake complexity, by presenting a randomized distributed (Monte Carlo) algorithm for MIS that, with high probability, computes an MIS in $O((\log\log{n})\log^*n)$ awake complexity and $O((\log^3 n) (\log \log n) \log^*n)$ round complexity. Our algorithms work in the $\CONGEST$ model where  messages of size $O(\log n)$ bits can be sent per edge per round.
\end{abstract}

\section{Introduction}\label{sec:introduction}

\subsection{Maximal Independent Set Problem}

Computing the \emph{maximal independent set} (\emph{MIS})  is one of the fundamental and most well-studied problems in distributed graph algorithms.  Given a graph with $n$ nodes, each node must (irrevocably) commit to being in a subset $M \subseteq V$ (called the MIS) or not such that (i) every node is either in $M$ or has a neighbor in $M$ and (ii) no two nodes in $M$ are adjacent to each other.

Because of the importance of MIS, distributed algorithms for MIS  have been studied extensively for
the last four decades mainly with a focus on improving the  time complexity
(i.e.,\ the number of rounds).
In 1986, Alon, Babai, and Itai \cite{Alon_1986} and Luby \cite{Luby_1986} presented a randomized distributed MIS algorithm  that takes  $O(\log{n})$  rounds ($n$ is the number of nodes in the graph) with high probability. Since these seminal results, there has been a lot of significant progress in recent years in designing progressively faster distributed MIS algorithms. For $n$-node graphs with maximum degree $\Delta$, Ghaffari \cite{Ghaffari_2016_SODA} presented a randomized MIS algorithm running in
    $O(\log \Delta) + 2^{O(\sqrt{\log\log n})}$ rounds,
improving over the algorithm of Barenboim, Elkin, Pettie and Schneider \cite{Barenboim_2016} that runs in
    $O(\log^2 \Delta) + 2^{O(\sqrt{\log\log n})}$ rounds. We note that these two results assume the $\LOCAL$ model. 
The run time was further improved by Rozhon and Ghaffari \cite[Corollary $3.2$]{Rozhon2020} to
    $O(\log{\Delta}  + \poly(\log{\log{n}}))$ rounds, which is currently the best known bound
    for randomized algorithms  in the $\LOCAL$ model.
The currently
best known randomized algorithm in the $\CONGEST$ model takes
$O(\log{\Delta}\log\log{n} + \poly(\log{\log{n}}))$ rounds~\cite{GGR20}.
Thus, the currently known
best algorithms of MIS (\cite{Rozhon2020, GGR20, Ghaffari_2016_SODA}) are dependent on
$\Delta$ (the maximum degree), and hence still take $O(\log n)$ rounds (even in the $\LOCAL$ model) for graphs with $O(\poly(n))$
degree. As far as deterministic algorithms are concerned, the best known algorithms take $O(\poly(\log n))$ rounds  in the $\LOCAL$ as well as $\CONGEST$ models \cite{Rozhon2020,GGR20}. 

There are faster distributed algorithms known for special classes of
graphs such as   trees~\cite{Ghaffari_2016_SODA,lenzen} and Erdos-Renyi random graphs~\cite{tcsgnp,Ghaffari_2016_SODA}, but they
still take $\Omega(\sqrt{\log{n}/\log\log{n}})$  rounds.
There are also MIS
algorithms that run faster on graphs with low arboricity, but they nevertheless
take $O(\log n)$ rounds on high arboricity graphs~\cite{Ghaffari_2016_SODA,barenboim}.

\begin{sloppypar}
   
While the above results  make significant progress in the round complexity of the MIS problem for some specific graphs, however,  in general graphs, the best known running time is still $O(\log{n})$ (even for randomized algorithms and even in the $\LOCAL$ model). Furthermore, there is a fundamental lower bound of
    $\Omega\left(\min\left\{\frac{\log \Delta}{\log \log \Delta}, \sqrt{\frac{\log n}{\log \log n}}\right\}\right)$ rounds
due to Kuhn, Moscibroda, and Wattenhofer~\cite{Kuhn_2016} that also applies to randomized algorithms and holds even in the $\LOCAL$ model. Thus, for example, say,  
when $\Delta = 2^{\Omega(\sqrt{\log n \log \log n})}$, 
it follows that one cannot hope for algorithms faster than $\sqrt{\log{n}/\log\log{n}}$ rounds. 
Balliu, Brandt, Hirvonen, Olivetti,  Rabie, and 
Suomela \cite{Balliu_2019}  showed  that one cannot hope for algorithms that run within
    $o(\Delta) + O(\log^*n)$ rounds
for the regimes where
    $\Delta \ll \log \log n$ (for randomized algorithms) \cite[Corollary $5$]{Balliu_2019}
and
    $\Delta \ll \log n$ (for deterministic algorithms) \cite[Corollary $6$]{Balliu_2019}.
(See also results on an improved lower bound \cite{MIS-trees-lower}.)

\end{sloppypar}

\subsection{Awake Complexity}

Energy is a premium resource in various settings such as  battery-powered wireless networks and sensor networks.
The bulk of the energy is used by the nodes (devices)  when they are {\em ``awake''}, i.e., when  they are  sending, receiving,  and even just listening for messages.
It is well-known that the energy used  by a node when it is idle and just listening (waiting to hear from a neighbor) is only slightly smaller than the energy used  in a transmitting or receiving state \cite{Zheng_2005, Feeney_2001}. 
On the other hand, the energy used in the ``sleeping'' state, i.e., when a node  has switched off its communication devices and is not sending, receiving or listening, is \emph{significantly less} than in the transmitting/receiving/idle (listening) state~\cite{Zheng_2005, Feeney_2001, King_2011, Wang_2006, Yang_2013}. A node may choose to enter and exit this sleeping mode in a judicious way to save energy during the course of an algorithm.\footnote{This has been exploited by  protocols to save power in ad hoc wireless networks by  switching between two states --- \emph{sleeping} and \emph{awake} --- as needed (the MAC layer provides support for switching between these states~\cite{Zheng_2005, Yang_2013, Murthy_2004_Book}).}

There has been a lot of recent theoretical interest in designing energy-efficient distributed algorithms for various fundamental problems such as maximal independent set, maximal matching, coloring, broadcasting, spanning tree construction, breadth-first tree construction, etc.\ (see e.g., \cite{energy1,podc2020,ghaffari-sleeping,BM21,energy2,energy3,energy4,AMP22}). 
These algorithms operate by  minimizing
the   number of rounds in which any node is  {\em awake},  also called the {\em awake complexity}.  
 An intriguing question is whether one can design
distributed algorithms for various problems that have significantly  smaller awake complexity compared to existing algorithms. However, this is challenging,  since a node can only communicate with a neighboring node that is awake (note  that a sleeping node  does not  send or receive messages and also messages sent to it are lost). As a result, coordinating (or scheduling) such communication in an efficient manner (without keeping any node awake for a long time)  becomes non-trivial.  

\subsection{Model and Complexity Measures}
\label{sec:model}

\noindent {\bf Distributed Computing Model.}
 We are given an anonymous distributed network of $n$ nodes, modeled as an undirected graph $G = (V,E)$. Each node hosts a processor with limited initial knowledge. 
 We assume that each node has ports (each port having a unique port number); each incident edge is connected to one distinct port. We assume that each node knows a common value $N$, a polynomial upper bound on $n$.

Nodes are allowed to communicate through the edges of the graph $G$ and it is assumed that communication is {\em synchronous} and occurs in rounds. In particular, we assume that each node knows the current round number (starting from round 0). In each round, each node can perform some local computation (which finishes in the same round) including accessing a private source of randomness, and can exchange messages of size $O(\log n)$ bits with each of its neighboring nodes. 

\begin{sloppypar}
This standard model of distributed computation is called the $\CONGEST$ model~\cite{Peleg_2000_Book}. We note that our algorithms also, obviously, apply to the $\LOCAL$ model, another standard model~\cite{Peleg_2000_Book} where there is no restriction on the size of the messages sent per edge per round. 
Though the  $\CONGEST$ and $\LOCAL$  models do not put any constraint on the computational power of the nodes,  our algorithms perform only light-weight computations  (each node processes only $\poly(\log n)$ bits per round and takes computation time essentially linear in the number of bits processed). 
\end{sloppypar}

 \smallskip
 
\noindent {\bf Sleeping Model.}  We assume the  sleeping model~\cite{podc2020}, where 
 a node can be in
 either of the two states --- sleeping or awake.  
 (At the beginning, we assume that all nodes are awake.)  This is a simple generalization of the standard distributed computing model, where nodes are always assumed 
 to be awake.
 In the sleeping model, each node decides to be either \textit{awake} or \textit{asleep} in each round (till it terminates), corresponding to whether the node can receive/send messages and perform computations in that round or not, respectively. That is, any node $v$ can decide to {\em sleep} starting at any (specified) round of its choice. We assume that all nodes know the correct round number whenever they are awake. A node can {\em wake up} again later at any {\em specified} round and enter the {\em awake} state.
 We note that the model allows a node to cycle through the process of sleeping  in some round and waking up at a later round as many times as it wants. To summarize, distributed computation in the sleeping model proceeds 
 in \emph{synchronous} rounds and each round consists of the following steps: (1) Each awake node can perform local computation. (2) Each awake node can send a message to its adjacent nodes.
    (3) Each awake node can receive messages sent to it in this round (in the previous step) by other awake nodes.

  In the sleeping model, let $A_v$ denote the number of awake rounds  for a node
    $v$ before it terminates (i.e., finishes the execution of the algorithm, locally). We define the \emph{(worst-case) awake complexity}
     as $\max_{v \in V}A_v$. For a randomized algorithm, $A_v$ will be a random variable
     and our goal is to obtain high probability bounds on the awake complexity.
     While the main goal is to reduce awake complexity, we also strive to minimize
     the  \emph{round complexity}, where both, sleeping and awake rounds, are counted.

     This paper assumes the sleeping model in the  $\CONGEST$ setting, which we call the SLEEPING-CONGEST model.
     
     Several recent works (see Section \ref{sec:related}) have designed    distributed algorithms in the SLEEPING-CONGEST  model for fundamental problems such as MIS, approximate  matching and vertex cover, spanning tree, minimum spanning tree, coloring, and other problems \cite{podc2020,ghaffari-sleeping,BM21,AMP22,ghaffari-podc2023}.

\subsection{Our Contributions}
\label{sec:results}

In light of the difficulty in breaking the $o(\log{n})$ round  barrier of MIS  and the lower bound of $\Omega(\sqrt{\log{n}/\log\log{n}})$ rounds in the standard model (that applies even for $\LOCAL$ algorithms), as well as motivated by resource considerations discussed above, a fundamental question that we address in this paper is:

\begin{center}
    \fbox
    {
        \begin{minipage}{20.5em}
            \emph{Can we design a distributed MIS algorithm with $o(\log n)$ awake complexity?}
        \end{minipage}
    }
\end{center}

We answer the above question in the affirmative and actually show a much stronger bound.
Our main contribution is that we show that MIS can be  computed in (worst-case) awake complexity of $O(\log \log n)$ rounds,  bypassing the $\Omega(\sqrt{\log{n}/\log\log{n}})$   lower bound on the round complexity in an exponentially better fashion.
Specifically, we present the following  results.
\begin{enumerate}
\item  We present a randomized distributed (Monte Carlo) algorithm for MIS that with high probability  computes an MIS and has $O(\log\log{n})$ awake complexity. This is the first  distributed MIS algorithm with $O(\log \log n)$ awake complexity. 
This algorithm has {\em round complexity} $O(\log^7 n \log \log n)$.
 Our bounds hold even in the $\CONGEST$ model where messages of $O(\log n)$ bits can be sent per edge per round.
\item We  show that we can  reduce the round complexity  at the cost of a slight increase in awake complexity by presenting a randomized MIS algorithm  with $O((\log \log n) \log^* n )$ awake complexity and  $O((\log^3 n) (\log \log n) \log^* n)$ round complexity in the $\CONGEST$ model.
\end{enumerate}

Our work answers a key question left open in \cite{podc2020}, namely whether one can design MIS algorithms with
(even) $o(\log n)$
(worst-case) awake complexity. We note that prior results \cite{podc2020,ghaffari-sleeping} presented algorithms in the sleeping
model with $O(\log n)$ awake complexity (see Section \ref{sec:related}).
Our results show that we can compute an MIS in an awake complexity that is \emph{exponentially} better compared to the best known  round complexity
of $O(\log n)$. 
Since a node spends a significant amount of energy  only in its awake rounds, our algorithms are highly energy-efficient compared to the existing algorithms.

\section{Related Work and Comparison}
\label{sec:related}

\noindent \textbf{Prior Work in the Sleeping Model for MIS.}
Chatterjee, Gmyr, and Pandurangan~\cite{podc2020} posit the sleeping model and showed that MIS in general graphs can be
solved in $O(1)$ \emph{node-averaged} awake complexity. Node-averaged
awake complexity is measured by the \emph{average} number of rounds a node is
awake. They also defined \emph{worst-case} awake complexity (used in this paper) which is the
worst-case number of rounds a node is awake until it finishes the algorithm. The
worst-case awake complexity of their MIS algorithm is $O(\log n)$, while the
worst-case complexity (that includes all rounds, sleeping and awake) is
$O(\log^{3.41}n)$ rounds. An important  question left open in \cite{podc2020} is whether one can design an MIS algorithm 
with  $o(\log n)$ worst-case awake complexity (even in the $\LOCAL$ model).\footnote{In
this paper, we do not focus on the node-averaged awake complexity measure, and only focus
on the (worst-case) awake complexity. However, it is fairly straightforward to augment the algorithms of this paper so that they  also give an $O(1)$-node averaged complexity in addition to their (worst-case) awake complexity guarantees by using the approach of \cite{podc2020}.} 

Ghaffari and Portmann~\cite{ghaffari-sleeping} subsequently improved the round complexity bound of Chatterjee, Gmyr, and Pandurangan~\cite{podc2020}. They present  a randomized MIS algorithm
that has worst-case awake complexity of $O(\log n)$, round complexity of $O(\log n)$, while having $O(1)$ node-averaged awake complexity (all
bounds hold with high probability). They also present algorithms for $(1+\epsilon)$ approximation of maximum matching and $(2+\epsilon)$ approximation of
minimum vertex cover with the same bounds on round complexity
and node-averaged awake complexity.
Hourani, Pandurangan, and Robinson \cite{hourani-mis}  presented
randomized MIS algorithms that have $O(\poly(\log \log n))$ awake complexity for certain special classes of {\em random graphs},
including random geometric graphs (of arbitrary dimension). These algorithms work only in the $\LOCAL$ model as
linear (in $n$) sized messages need to be sent per edge per round. This result is subsumed by the results of the current paper.

Subsequent to the publication of the arxiv version of our paper \cite{DMP22}, Ghaffari and Portmann \cite{ghaffari-podc2023}
presented a distributed  MIS algorithm with the same worst-case awake complexity of $O(\log \log n)$, but a better round complexity of round complexity $O(\log^2 n)$. 
They also present a second algorithm that  has worse  awake complexity of $O((\log \log n)^2)$ but with  a better round complexity of $O(\log n \log \log n \log^*n)$.

~\\\noindent \textbf{Other Works in the Sleeping Model.}
Barenboim and Maimon~\cite{BM21} showed that many problems, including broadcast, construction of a spanning tree,
 and leader election can be solved deterministically in $O(\log n)$ awake complexity in the sleeping model.
 They construct a spanning tree called Distributed Layered Tree (DLT) in $O(\log n)$ awake complexity deterministically. In this tree, broadcast and convergecast can be accomplished in $O(1)$ awake rounds.
 They also showed
 that fundamental symmetry breaking problems such as MIS and ($\Delta+1$)-coloring can be solved deterministically
 in $O(\log \Delta + \log^*n)$ awake rounds in the $\LOCAL$ model, where $\Delta$ is the maximum degree. 
 (Note that, in general,  this can take $O(\log n)$ awake rounds when $\Delta = \Omega(\poly n)$.) 
 Their algorithm only works in the $\LOCAL$ model (as opposed to the $\CONGEST$ model), as it needs large-sized (polynomial number of bits) messages to be sent over an edge.
 They also define  the class of {\em O-LOCAL} problems (that includes MIS and coloring), where such a problem is one that can be solved sequentially according to some acyclic orientation of the edges of the input graph where the decision of a node depends on the decisions of the nodes in the subtree rooted at it. 
 They 
showed that problems in this class admit  a deterministic algorithm that runs in $O(\log \Delta)$ awake time and $O(\Delta^2)$ round complexity.
 Maimon~\cite{maimon} presents trade-offs between awake and round complexity for O-LOCAL problems.

Augustine, Moses Jr., and Pandurangan~\cite{AMP22} give an $O(\log n)$ awake complexity algorithm for the minimum spanning tree (MST) problem (in the CONGEST model). They use a spanning tree construction called the Labelled Distance Tree (LDT) which we also use in our algorithm.

~\\ \noindent \textbf{Other Prior Works on Distributed Energy-Efficient Algorithms.}

There has been significant research on energy-efficient distributed algorithms over the years --- more than we can survey here --- and we restrict ourselves to those that are
most relevant.  

A first, closely linked line of work is that of Chang, Kopelowitz, Pettie, Wang, and Zhan~\cite{energy1} on radio networks (inspired by earlier work on energy efficient algorithms in radio networks e.g., \cite{nakano,jurdin,pajak}).
Unlike our paper (and the other papers mentioned in the previous paragraph) that use the \emph{SLEEPING-CONGEST} model, this line of work assumes a model with additional communication restrictions that reflect the behavior of radio networks; we refer to it as the \emph{SLEEPING-RADIO} model. More concretely, in this model, nodes can only broadcast messages; hence the same message is sent to all neighbors. Additionally, collisions can occur at a listening node if two neighboring nodes transmit simultaneously in the same round. There are a few variants depending on how collisions are handled.) In this model, they study the {\em energy complexity} measure, which is defined identically to (worst-case) awake complexity (i.e., both measures count only the rounds that a node is awake. 

Energy-efficient algorithms for several
problems  such as broadcast, leader election, breadth-first search, maximal matching, diameter and minimum-cut computation have been studied in the SLEEPING-RADIO model~\cite{energy1,energy2,energy3,energy4,energy5, CS22, CDJ23,Chang23}. An interesting open problem
is whether our sub-logarithmic bounds on MIS awake complexity in the SLEEPING-CONGEST can be obtained 
in the SLEEPING-RADIO model.

King, Phillips, Saia, and Young~\cite{King_2011} also study a similar model  where nodes can be in
two states: sleeping or awake (listening and/or sending). They present an
energy-efficient algorithm in this model to solve a reliable broadcast problem.  We also refer to  the literature 
on resource competitive algorithms where there is limited energy available both to the algorithm and the
(jamming) adversary (e.g., \cite{young1,young2,young3}).

Finally, Fraigniaud, Montealegre, Rapaport and Todinca \cite{FMRT23} consider energy-efficiency from a slightly different angle. They do not assume nodes have a sleeping capability, but instead seek to minimize the maximum, taken over all nodes (or all edges), of the number of rounds in which these nodes (or edges) are active in sending messages (or transmitting messages). These two complexity measures are called respectively node-activation and edge-activation complexity, and the authors show that every Turing-computable problem can be solved with $O(1)$ node- and edge-activation complexity in the CONGEST setting, at the cost of possibly exponential in $n$ round complexity. In the LOCAL setting, the round complexity can be reduced to polynomial in $n$.  
\section{High-level Overview and Techniques}
The best known distributed MIS algorithms (\cite{Rozhon2020, GGR20, Ghaffari_2016_SODA}) in the traditional setting suffer from a $\log \Delta$ dependency in the round complexity, where $\Delta$ is the maximum degree
(see Section \ref{sec:introduction}). Prior to this work, that was also the case in the sleeping model as well.\footnote{In particular,
the algorithms of \cite{podc2020, ghaffari-sleeping} which had  optimal $O(1)$ rounds {\em node-averaged} awake complexity, however, 
had $O(\log n)$ (worst-case) awake complexity.}
Rather than improve these algorithms to remove this dependency (which appears very difficult), we improve a different algorithm: the well-known {\em randomized greedy MIS} algorithm \cite{Coppersmith_1989,Blelloch_2012, Fischer_2018}, a variant of Luby's algorithm \cite{Luby_1986}.

The (sequential) randomized greedy MIS algorithm considers nodes (of some graph $G=(V,E)$) in random order and adds them to the output set unless one of their neighbors is already in it. If $v_1,\ldots,v_n$ is the random node ordering considered by the algorithm, then it is well-known that the output is the lexicographically first MIS (LFMIS) \cite{Coppersmith_1989} with respect to that (random) ordering.\footnote{Given two (MIS) subsets $X \neq Y$ of $V$, such that $X \nsubseteq Y$ and $Y \nsubseteq X$, $X$ is lexicographically smaller (with respect to that ordering) than $Y$ if and only if the smallest differing element between $X$ and $Y$ is in $X$.} 
Fischer and Noever \cite{Fischer_2018} showed that
the randomized greedy MIS can be implemented in the (traditional) distributed computing model in $O(\log n)$ rounds with high probability and also that this bound is tight.
On the other hand, we show that randomized greedy MIS can be implemented in $O(\log \log n)$ awake complexity. We build to this result in three steps. 
\smallskip

\noindent {\bf Algorithm \AuxMISone{}.} First, we give a simple awake-efficient variant (Algorithm \AuxMISone{} in Subsection \ref{subsec:auxMIS}) of the naive distributed implementation of the above (sequential) algorithm. Let $I$ be an upper bound on the randomly chosen IDs. Then, the naive distributed implementation works as follows. In each round $i \in [1,I]$, all nodes transmit whether they have joined the MIS or not to their neighbors. After which, the node with ID $i$ (if it exists) enters the MIS if none of its neighbors already have. Clearly, the naive implementation has an excessive $O(I)$ awake complexity. However, we can reduce the awake complexity exponentially, that is, to $O(\log I)$. 

To reduce the awake complexity, we use a \emph{virtual binary tree structure} (see Subsection \ref{subsec:virtualBinaryTree}), similar to that of \cite{BM21}, to carefully coordinate the communication between the nodes. More precisely, a (virtual) tree with $t$ leaves (where the same tree is locally determined by each node using the parameter $t = 2^{\lceil \log I \rceil}$) tells each node in which round to be awake and communicate to its neighbors whether it is in the MIS or not. This virtual tree ensures that for any two neighboring nodes $v,v'$ with $id_v <  id_{v'}$, $v$ and $v'$ are both awake in some round $i$ that satisfies $id_{v} < i \leq id_{v'}$. As a result, a node needs to be awake in only $O(\log I)$ well-chosen rounds (where $\log t = O(\log I)$ is the depth of the virtual tree) to correctly implement randomized greedy MIS. 
This virtual tree coordination framework is reused in our third algorithm, \AwakeMIS{}, and we believe it can be useful in general for designing awake-efficient algorithms.

\smallskip

\noindent {\bf Algorithm \AuxMIStwo{}.} Second, we give a more awake-efficient variant (Algorithm \AuxMIStwo{} in Subsection \ref{subsec:auxMIS}) with an improved dependency on the ID upper bound $I$. This improved dependency comes into play when $I$ is super-polynomial (or even worse) in the number of nodes $n$. Having good awake complexity in this particular scenario, which happens in \AwakeMIS{}, is crucial for our $O(\log \log n)$ awake complexity main result.

To obtain an improved dependency on $I$, we use a spanning tree structure, called a \emph{labeled distance tree} (LDT), introduced in \cite{AMP22} (an improvement on a similar structure, introduced previously in \cite{BM21}). Crucially, these trees can be used to broadcast and rank nodes (i.e, assign IDs in $[1,n]$) in $O(1)$ awake complexity, while the LDT itself can be constructed in $O(\log n)$ awake complexity deterministically \cite{AMP22}. Hence, one can first build a LDT and rank the nodes in $O(\log n)$ awake complexity. Since nodes are ranked arbitrarily, we can then have the root broadcast a uniformly random permutation of $[1,n]$ in $O((n \log n)/\log I)$ consecutive broadcasts (recall that messages can be of size $O(\log I) = O(\log n)$ bits, and thus can be sent in $\CONGEST$). 
Consequently, nodes obtain new IDs in $[1,n]$ that correspond to a uniformly random ordering. It remains only to run Algorithm \AuxMISone{}, which now takes $O(\log n)$ awake complexity (due to the smaller IDs).
\smallskip

\noindent {\bf Algorithm \AwakeMIS{}.} For our main result, we use the previously described techniques as well as the following two key properties of the \emph{randomized greedy MIS algorithm}. 
The first is the \emph{composability property}. One can run the randomized greedy MIS algorithm on the first $k > 0$ nodes, then run that algorithm again but on the remaining nodes, that is, those which are not neighbors of the first computed MIS. The union of the two computed MIS's is the output of the randomized greedy MIS algorithm on $G$.
The second is the \emph{residual sparsity property} --- stated formally in Lemma \ref{lem:sparsityLemmaForRandomOrder} in Subsection \ref{subsec:randomizedGreedyMISProperty} --- which shows that after processing the first $k$ nodes in the random ordering, the degree of the residual graph (i.e., the subgraph induced by the rest of the nodes minus the neighbors of MIS nodes among the first $k$ nodes) is reduced (essentially) to $O(n/k)$ with high probability. 

In Algorithm \AwakeMIS{} (described in Section \ref{sec:MIS}), nodes are partitioned into $O(\log^2 n)$ batches. More precisely, each node picks a pair $(i,j) \in [1,\ell] \times [1,2\Delta']$ with some well-chosen probabilities, where $\ell = O(\log n)$ and $\Delta'= O(\log n)$ are two parameters. To compute the MIS, we consider batches sequentially in phases (according to the lexicographical ordering). During the first ``communication'' round of each phase, nodes inform their neighbors whether they are already in the MIS or not. Moreover, just as in \AuxMISone{}, nodes use a virtual binary tree structure to coordinate in which of these rounds to be awake or to sleep. (Hence, any given node is awake for at most $O(\log \log n)$ communication rounds.) For the remaining rounds of some phase $(i,j)$, nodes of batch $(i,j)$ with no neighbors already in the MIS run Algorithm \AuxMIStwo{} to compute an MIS over the batch's nodes. Crucially, we show that the subgraph induced by the nodes running Algorithm \AuxMIStwo{} is \emph{shattered}: that is, it consists only of $O(\log n)$-sized components with high probability. Hence, nodes run \AuxMIStwo{} using $O(\log \log n)$ awake complexity only. (Here, it is important that the second term of \AuxMIStwo{}, caused by the $\CONGEST$ bandwidth, adds up to $O(\log \log n)$.)  From this, it is easy to see that Algorithm \AwakeMIS{} has $O(\log \log n)$ awake complexity.

However, how do we ensure that the subgraph induced by the nodes running Algorithm \AuxMIStwo{} is shattered? First, we adjust the probabilities that nodes choose a given (batch) pair to ensure that the first $2 \Delta'$ batches contain with high probability half as many nodes as the next $2 \Delta'$, and so on. Hence, by the residual sparsity property, the subgraph induced by the nodes (with no MIS neighbors) within any of these collections of $2 \Delta'$ batches has small $O(\log n)$ degree. (In fact, we must first use the composability property to show that the MIS computed throughout all previous phases is the output of randomized greedy MIS on the nodes in all previous batches.) Furthermore, nodes within any such collection independently and uniformly chose any of the $2 \Delta'$ batches. Hence, for each node, the expected number of neighbors (themselves with no MIS neighbors) is less than $1/2$. In this case, a simple branching process argument --- stated formally in Lemma \ref{lem:branchingProcess} in Subsection \ref{subsec:graphShattering} --- shows that the subgraph induced by a given batch's nodes (with no MIS neighbors) is shattered.

\section{Preliminaries: Notation and Randomized Techniques}\label{sec:randomizedprelim}
\subsection{Notation}
For any subset $V' \subseteq V$, let $G[V']$ denote the subgraph of $G$ induced by $V'$.
For any node $v$, let $N(v)$ denote the union of $v$ and the set of its neighbors. Similarly, for any set of vertices $V' \subseteq V$, let $N(V')$ denote the union of $V'$ and the set of neighbors of any node in $V'$. 
For any two integers $i$ and $j$, $[i,j]$ is used to denote the set $\{i,\ldots,j\}$. 

\subsection{Chernoff Bounds}
The following Chernoff bounds~\cite{Upfalbook} are used in the later sections, where the second bound is obtained by applying the inequality $\ln(1+\delta) \geq (2\delta)/(2+\delta)$ to Theorem 4.4 (Inequality 1) in~\cite{Upfalbook}.

\begin{lemma}
\label{lem:chernoffBound}
Let $X_1,\ldots,X_k$ be independent Bernoulli random variables with parameter $p$. Then, 
\begin{itemize}
    \item For any $0 \leq \delta \leq 1$, $\Pr[\sum_{i=1}^k X_i \leq (1-\delta) pk] \leq e^{-\frac{\delta^2 kp}{2}}$,
    \item For any $\delta \geq 0$, $\Pr[\sum_{i=1}^k X_i \geq (1+\delta) pk] \leq e^{-\frac{\delta^2 kp}{2+\delta}}$.
\end{itemize}
\end{lemma}

\subsection{Sequential Randomized Greedy MIS} 
\label{subsec:randomizedGreedyMISProperty}

The \emph{sequential randomized greedy MIS} algorithm processes each node in a sequential but random order. Each node is added to the output set if it is not a neighbor of a node already in that set. It is well-known that this algorithm outputs the lexicographically first MIS (LFMIS), with respect to the random node ordering.

Given a random node ordering $v_1, v_2, \ldots, v_n$, Lemma \ref{lem:sparsityLemmaForRandomOrder} --- a slight generalization of Lemma 1 in \cite{K18} --- shows that after the first $t$ nodes (according to the node ordering) are processed by the sequential randomized greedy order MIS, the maximum degree of the subgraph induced by the remaining nodes among the first $t' > t$ nodes (those which have not been added, nor are neighbors of an already added node) has decreased (almost) linearly in $t$. 
In fact, the lemma is more general. For example, it applies to distributed algorithms that compute the LFMIS over the subgraph induced by the first $t$ nodes, according to the random node ordering mentioned above. 

\begin{lemma}
\label{lem:sparsityLemmaForRandomOrder}
Let $t, t'$ be two integers such that $1 \leq t < t' \leq n$. Let $V_t$ denote the (set of the) first $t$ nodes, $V_{t'}$ the (set of the) first $t'$ nodes and $M_t$ the LFMIS over $G[V_t]$. Then, for any constant $\epsilon > 0$, $G[V_{t'} \setminus N(M_t)]$ has maximum degree at most $\frac{t'}{t}  \ln (n/\epsilon)$ with probability at least $1-\epsilon$.
\end{lemma}

\begin{proof}
Note that $(V_{t'} \setminus N(M_t)) \subseteq \{v_{t+1},\ldots,v_{t'}\}$. We show that, with probability at least $1-\epsilon/n$, for any $j \in [t+1,t']$, either $v_j$ has degree at most $\frac{t'}{t}  \ln (n/\epsilon)$ in $G[V_{t'} \setminus N(M_t)]$ or $v_j \in N(M_t)$. 
The lemma statement holds by a union bound (over $j$).

Let $j \in [t+1,t']$. 
We apply the principle of deferred decisions \cite{Upfalbook}. 
More precisely, we first fix (the random choice of) which node is in position $j$ --- that is, $v_j$. After which, we fix (the random choices of) which nodes are in position 1 to $t$ sequentially --- that is, $v_1$ to $v_t$.  
For any integer $i \in [1,t]$, let $V_i$ denote the first $i$ (fixed) nodes and $M_i$ be the LFMIS over $G[V_i]$. Additionally, let $U_i = (N(v_j) \cap V_{t'}) \setminus N(M_{i-1})$ and $d_i = |U_i|$. Then, $\Pr[v_i \in U_i \; | \; v_j,v_1,\ldots,v_{i-1}] \geq \frac{d_i}{t'-1-(i-1)} \geq \frac{d_i}{t'}$.

The sequence $(d_i)_{i \in [1,t]}$ is decreasing. If $d_t \leq \frac{t'}{t} \ln (n/\epsilon)$, then $v_j$ has degree at most $  \frac{t'}{t} \ln (n/\epsilon)$ in $G[V_{t'} \setminus N(M_t)]$. Otherwise, $d_t > \frac{t'}{t} \ln (n/\epsilon)$. 
Then, $\Pr[\forall i \leq t, v_i \notin U_i \; | \; v_j] \leq \prod_{i=1}^t \Pr[v_i \notin U_i \; | \; v_j,v_1,\ldots,v_{i-1}] \leq \prod_{i=1}^t (1-\frac{d_i}{t'}) \leq (1-\frac{d_t}{t'})^t \leq e^{-\frac{d_t}{t'} t} \leq e^{- \ln (n/\epsilon)} \leq \epsilon/n$. 
In other words, there exists $i \in [1,t]$ such that $v_i \in U_i$ with probability at least $1 - \epsilon/n$. In which case, since $M_i$ is the LFMIS over $G[V_i]$, $v_i \in M_i$. Thus, $v_i \in M_t$ and $v_j \in N(M_t)$ (with probability at least $1 - \epsilon/n$).
\end{proof}

\subsection{Simple Graph Shattering}
\label{subsec:graphShattering}

Consider some $n$-node graph $H$ with maximum degree $\Delta$ and partition the nodes into $2 \Delta$ sets uniformly at random. More precisely, each node is in set $U_j$, for any $j \in [1,2\Delta]$, with probability $1/ (2\Delta)$. Then, a simple branching process argument implies that the corresponding induced subgraphs $H[U_j]$ are ``shattered'': that is, they are composed of small $O(\log n)$-sized connected components with high probability.

\begin{lemma}
\label{lem:branchingProcess}
For any $j \in [1,2 \Delta]$, the connected components of $H[U_j]$ are of size at most $6 \ln (n/\epsilon)$ with probability at least $1 - \epsilon$.
\end{lemma}

\begin{proof}
For any $j \in [1,2 \Delta]$, consider some node $v \in U_j$. (If $|U_j| = 0$, the lemma statement obviously holds.) We assume, by the principle of deferred decisions, that $v$ is the only initially revealed node of $U_j$ (and for all other nodes, we do not know whether they are in $U_j$ or not) and we find out which nodes are in $U_j$ through a BFS search (over $H[U_j]$ only) starting at $v$. In more detail, the BFS queue initially consists only of $v$. In the first step, $v$ is dequeued and for each unrevealed neighbor $w \in N(v)$, we reveal whether $w$ is in $U_j$. For each such $w \in N(v)$, if $w \in U_j$ then $w$ is added to the queue. Once all neighbors have been revealed, the first step is done. Subsequent steps are executed similarly, but with that step's dequeued node, until the queue is empty. Importantly, the number of steps executed before the queue is empty is the size $C(v)$ of the connected component of $H[U_j]$ containing $v$. Moreover, it is a random variable that depends on each revealed node --- that is, whether that revealed node is in $U_j$ or not. Finally, each unrevealed node $w$, once revealed, is in $U_j$ with probability at most $\frac{1}{2 \Delta}$. 

The above (BFS search) randomized process is hard to analyze since a node dequeued in some step $k$ might have neighbors that were revealed in previous steps. Instead, we consider an easier-to-analyze but related randomized process: BFS on a branching process. 
Let the number of nodes in the queue at the start of step $k \geq 0$ be denoted by $A_t$. Initially, there is a single node in the queue --- that is, $A_0 = 1$. Subsequently, for any step $k \geq 1$, $A_k = A_{k-1} - 1 + Y_k$, where the random variables $(Y_k)_{k \geq 1}$ are independent and binomially distributed with parameters $\Delta$, the total number of events, and $\frac{1}{2 \Delta}$, the probability of each event being successful.
Then, the \emph{size} of this randomized process is $C' = \min\{k \geq 0 \; | \; A_t = 0\}$. Importantly, $C'$ dominates the random variable $C(v)$ --- that is, $\Pr[C(v) \geq k] \leq \Pr[C' \geq k]$ for any positive integer $k$. By the definition of $C'$, $\Pr[C' > k] = \Pr[A_1 > 0,\ldots,A_k > 0] \leq \Pr[A_k > 0] \leq \Pr[A_0 + \sum_{i=1}^k Y_i > k] = \Pr[\sum_{i=1}^k Y_i \geq k]$. Note that since the $Y_i$ are independent binomially distributed random variables (with parameters $\Delta$ and $\frac{1}{2 \Delta}$), $\sum_{i=1}^k Y_i$ is the sum of $k \Delta$ Bernoulli random variables with parameter $\frac{1}{2 \Delta}$ and thus $\mathbb{E}[\sum_{i=1}^k Y_i] = k/2$. Hence, by the Chernoff bound (see Lemma \ref{lem:chernoffBound}), $\Pr[\sum_{i=1}^k Y_i \geq k)] \leq e^{-k/6}$. Consequently, $\Pr[C' \geq 6 \ln (n/\epsilon)] \leq \epsilon/n$ for any constant $\epsilon > 0$. Since $C'$ dominates $C(v)$, $\Pr[C(v) \geq 6 \ln (n/\epsilon)] \leq \epsilon/n$. By union bound over all connected components of $H[U_j]$ (of which there are at most $n$), the lemma statement holds.
\end{proof}

\section{Auxiliary Procedures}\label{sec:prelim}

\subsection{The Virtual Binary Tree Technique}
\label{subsec:virtualBinaryTree}

We provide a virtual binary tree construction similar to that in \cite{BM21}.
Let $i$ be an integer, provided as a parameter. The virtual (full) binary tree $\mathcal{B}([1,i])$ has depth $d = \lceil \log i \rceil$ and thus $y = 2^{d + 1} - 1$ nodes. These nodes are labeled with integers in $[1,y]$ according to an in-order tree traversal. 
See Figure 1 (left).
Given $\mathcal{B}([1,i])$, we can define the more convenient node-labeled (full) binary tree $\mathcal{B}^*([1,i])$ as follows. The tree structure is the same, but the node labels of $\mathcal{B}^*([1,i])$ are obtained by applying $g(x) = \lfloor x/2 \rfloor + 1$ to the node labels of $\mathcal{B}([1,i])$. Using $\mathcal{B}^*([1,i])$, we define, for any integer $k \in [1,i]$, a \emph{communication set} $S_k([1,i])$ of $d$ integers in $[1,i]$, as follows: $S_k([1,i])$ consists of the labels of all ancestors of the leaf node labeled $k$ in  $\mathcal{B}^*([1,i])$. 
See Figure 1 (right).

\begin{figure}[ht]
\footnotesize
\begin{subfigure}{.4\textwidth}
    \centering
    \begin{forest}
    [8,circle [4,circle,minimum size=16pt [2,circle,minimum size=16pt [1] [3]] [6,circle,minimum size=16pt [5] [7] ]][12,circle [10,circle [9] [11]][14,circle [13][15]]]]
    \end{forest}
    \label{fig:virtualBinaryTree1}
\end{subfigure}%
\hfill
\begin{subfigure}{.4\textwidth}
    \centering
    \begin{forest}
    [5,rectangle,draw,minimum size=16pt [3,circle,dotted,draw,minimum size=16pt [2,circle,minimum size=16pt [1] [2]] [4,circle,dotted,draw,minimum size=16pt [3] [4] ]][7,circle,draw,minimum size=16pt [6,circle,draw,minimum size=16pt [5] [6]][8,circle,minimum size=16pt [7][8]]]]
    \end{forest}
    \label{fig:virtualBinaryTree2}
\end{subfigure}%
\caption{Binary tree $\mathcal{B}([1,6])$ on the left and binary tree $\mathcal{B}^*([1,6])$ on the right. $S_3([1,6])$ consists of the dotted circle and rectangle nodes' labels and $S_5([1,6])$ of the (non-dotted) circle and rectangle nodes' labels. Note that $5 \in S_3([1,6]) \cap S_5([1,6])$ and $3 < 5 \leq 6$.}
\end{figure}
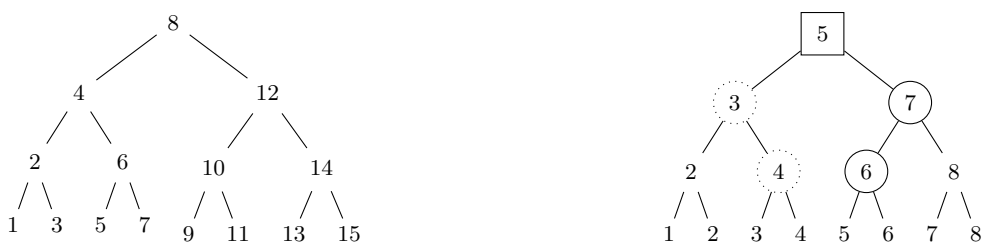

These sets have the following property (see Observation \ref{obs:communicationSet} below): for any integers $k, k' \in [1,i]$ such that $k < k'$, there exists an integer $r$ in both $S_k([1,i])$ and $S_{k'}([1,i])$ such that $k < r \leq k'$. 
Informally, we later use the sets $S_k([1,i])$ to decide when nodes with IDs in $[1,i]$ are awake or asleep. The above property allows us to decide, for any two nodes, on a common round in which they are guaranteed to be awake simultaneously (and thus communicate with each other).
Remember that in rounds in which nodes are not awake simultaneously, any message sent between two neighboring nodes is {\em lost if any one of them is asleep}.  

\begin{observation}
\label{obs:sizeCommunicationSet}
For any positive integers $k,i$ such that $1 \leq k \leq i$, $|S_k([1,i])| \leq \lceil \log i \rceil$. 
\end{observation}

\begin{observation}
\label{obs:communicationSet}
For any positive integers $k, k',i$ such that $1 \leq k < k' \leq i$, there exists an integer $r \in S_k([1,i]) \cap S_{k'}([1,i])$ such that $k < r \leq k'$. 
\end{observation}

\begin{proof}
Let the lowest common ancestor node in $\mathcal{B}^*([1,i])$ of the leaves labeled $k$ and $k'$ be labeled with $r'$. Then, according to the definition of a communication set, $r'$ is in both $S_{k}([1,i])$ and $S_{k'}([1,i])$.
Moreover, note that the corresponding nodes in $\mathcal{B}([1,i])$ are the internal (lowest common ancestor) node labeled $2(r'-1)$ and the leaf nodes labeled $2k - 1$ and $2 k' - 1$. By the property of the in-order tree traversal, $2k - 1 < 2(r'-1) < 2 k' - 1$. Hence, $k < r'-1/2 <  k'$. Since $r'$ and $k'$ are integers, $k < r' \leq  k'$.
\end{proof}

\subsection{Labeled Distance Trees}

For any $V' \subseteq V$ where $G[V']$ is connected, a \emph{labeled distance tree} (introduced in~\cite{AMP22}) is an oriented node labeled spanning tree over $G[V']$ rooted at some node $r \in V'$. Each node's label is its distance to $r$ in the spanning tree. (Note that the distance to $r$ in the spanning tree may be much larger than that in $G$.) In the distributed implementation, the LDT satisfies the following properties: (i) all nodes in the tree know the ID of $r$, called the ID of the LDT, (ii) each node knows its depth in the tree (i.e., the hop-distance from itself to the root of the tree via tree edges), and (iii) each node knows the IDs of its parent and children, if any, in the LDT. If a given graph is disconnected, one can compute a disjoint set of such LDTs, one per connected component, which we refer to as a \textit{forest of labeled distance trees (FLDT)}. 

 Multiple distributed LDT construction algorithms (both randomized and deterministic) are presented in \cite{AMP22}. We are interested in their two deterministic construction algorithms, which give different guarantees. The first lemma below captures the worst-case awake and round complexities of their first construction algorithm: Algorithm {\LDTCONSTRUCTAWAKE}. (Here, we use the improved version of Theorem~4 in~\cite{AMP22} which can be found in the arXiv version.) On the other hand, the second lemma captures the properties of a distributed LDT construction algorithm (Algorithm {\LDTCONSTRUCTROUND}) with faster round complexity but larger awake complexity; Corollary 1 of~\cite{AMP22} states such an algorithm exists, and we give a full construction procedure in Appendix~\ref{sec:app-ldt}.

\begin{lemma}[Theorem 4, \cite{AMP22}]
\label{lem:LDT-slow-round-construction}
For any connected $V' \subseteq V$ of at most $n'$ nodes with unique IDs in $[1,I]$, where $n'$ and $I$ are known to all nodes, $\LDTCONSTRUCTAWAKE$ deterministically constructs an LDT over $G[V']$ with $O(\log n')$ awake complexity, $O(n' (\log n') \log^4 I)$ round complexity and $O(\log I)$ bit messages. 
\end{lemma}

\begin{lemma}
\label{lem:ldt-fast-round-construction}
For any connected $V' \subseteq V$ of at most $n'$ nodes, where $n'$ is known to all nodes, with unique IDs in $[1,I]$, $\LDTCONSTRUCTROUND$ deterministically constructs an LDT over $G[V']$ with $O((\log n') \log^* I)$ awake complexity, $O(n' (\log n') \log^* I)$ round complexity and $O(\log I)$ bit messages. 
\end{lemma}

Next, we define two operations over a labeled distance tree.
\begin{definition}\label{def:ldt-ops}
For any connected $V' \subseteq V$ and an LDT spanning $V' \subseteq V$:
\begin{itemize}
    \item The \emph{broadcast} operation consists of sending the root's (input) message, denoted by $m_r$, to all of the LDT's nodes.
    \item The \emph{ranking} operation consists of having the nodes of $V'$ compute a total ordering --- more precisely, each node knows its rank in the ordering --- as well as the size $|V'|$. 
\end{itemize} 
\end{definition}

\cite{AMP22} provides a distributed algorithm for broadcasting over an LDT (see Observation 2 in~\cite{AMP22}). A distributed algorithm for ranking over an LDT is presented in Appendix~\ref{sec:app-ldt}.
The properties of these two algorithms are captured by the following lemma.

\begin{lemma}
\label{lem:LDT-operation}
For any LDT over at most $n'$ nodes, where $n'$ is known to all nodes, with unique IDs in $[1,I]$,
broadcast and ranking can be executed deterministically with $O(1)$ awake complexity, $O(n')$ round complexity and the size of the messages are $O(|m_r|)$ bits and $O(\log I)$ bits respectively. 
\end{lemma}

\subsection{Simple Awake-Efficient MIS Algorithms}
\label{subsec:auxMIS}

We provide two simple deterministic distributed lexicographically first MIS (LFMIS) algorithms with good awake complexity. Note that it is important for our main result (see Section \ref{sec:MIS}) that these algorithms compute the LFMIS rather than any arbitrary MIS.

The first algorithm (\AuxMISone{}) runs in $O(\log I)$ awake time, where $I$ is an upper bound on the nodes' IDs, and illustrates how to use the virtual binary tree technique (see Section \ref{subsec:virtualBinaryTree}). At a high-level, \AuxMISone{} is an awake-efficient version of the naive distributed implementation of sequential greedy MIS. Recall that the naive algorithm runs in $I$ rounds, and in the $i$th round, the node with ID $i$ (if it exists) communicates with all neighbors to check if any neighbor with smaller ID is already in the MIS. If not, it joins the MIS. Although \AuxMISone{} also has poor $O(I)$ round complexity, its $O(\log I)$ awake complexity is exponentially smaller. 

The second algorithm (\AuxMIStwo{}) runs in $O(\log n')$ awake time, where $n'$ is the number of nodes. Of particular interest to us is when $\log n'$ may be much smaller than the length of the node's IDs. This naturally happens if \AuxMIStwo{} is run on a subgraph obtained after shattering a much larger graph (e.g., with exponentially more nodes $n$) whose nodes had unique IDs. At a high-level, we compute labeled distance trees, then the root computes $n'$ and assigns uniformly random and small enough $O(\log n')$-length IDs to each node, all in $O(\log n')$ awake time. All nodes then compute the LFMIS using \AuxMISone{} and these new IDs in $O(\log n')$ awake time.

\smallskip
\noindent {\bf \AuxMISone{}: MIS in $O(\log I)$ awake time.} We assume nodes are given unique IDs in $[1,I]$, for some known value $I$. Initially, each node starts in the undecided state and computes the virtual binary tree $\mathcal{B}^*([1,I])$ locally (for the description, see Subsection \ref{subsec:virtualBinaryTree}). Then, nodes compute the LFMIS in $I$ rounds. In round $r$, the node that has ID $r$ as well as all nodes $u$ for which $r \in S_{id_u}([1,I])$ wake up --- where $S_{id_u}([1,I])$ is the \emph{communication set} defined using $\mathcal{B}^*([1,I])$, see Figure~\ref{fig:figureForVTMIS} for an example. Intuitively, the additional nodes wake up to communicate information on their state (i.e., undecided, in MIS or not in MIS), and the use of communication sets results in low awake time. All awake nodes send their state to all neighbors. Any awake undecided node that learns one of its neighbors is  in the MIS sets its state to ``not in MIS''. If the node with ID $r$ remains undecided despite the received messages, then it joins the MIS and sets its state to ``in MIS''.

\begin{lemma}
\label{lem:awakeDistributedGreedyMIS}
\AuxMISone{} computes the LFMIS with $O(\log I)$ awake complexity, $O(I)$ round complexity and $O(\log I)$ bit messages.
\end{lemma}

\begin{proof}
We first prove the correctness. The main difference with the naive distributed implementation of sequential greedy MIS is that here, in any round $r \in [1,I]$, not all nodes may be awake and thus the node $v$ with ID $r$ may be unaware that one of its neighbors is already in the MIS. However, by Observation \ref{obs:communicationSet}, there exists for any node $u$ with ID $r' < r$, a round $r^*$ such that $r' < r^* \leq r$ and both $u$ and $v$ are awake in $r^*$. Since $u$ can only decide to join the MIS in round $r'$ by the algorithm's definition, $u$ successfully communicates whether it joins the MIS or not to $v$ in round $r^*$. It follows that \AuxMISone{} implements sequential greedy MIS, and thus correctly computes an LFMIS.

As for the awake complexity, note that by Observation \ref{obs:sizeCommunicationSet}, each communication set of $\mathcal{B}^*([1,I])$ is of size $O(\log I)$. Since each node $u$ is only awake in rounds $r \in S_{id_u}([1,I])$, each node is awake for $O(\log I)$ rounds.
\end{proof}

\begin{figure}[ht]
\footnotesize
\begin{subfigure}{.48\textwidth}
    \centering
    \begin{forest}
    [5,circle,draw [3,circle,draw [2,circle [1] [2]] [4,circle,draw [3] [4] ]][7,circle [6,circle [5] [6]][8,circle [7][8]]]]    \end{forest}
    \caption{$S_3([1,6])$ consists of the circle nodes' labels.}
    \label{fig:virtualBinaryTree5}
\end{subfigure}\hfill%
\begin{subfigure}{.48\textwidth}
    \centering
    \begin{forest}
    [5,circle,draw [3,circle [2,circle [1] [2]] [4,circle [3] [4] ]][7,circle,draw [6,circle,draw [5] [6]][8,circle [7][8]]]]
    \end{forest}
    \caption{$S_5([1,6])$ consists of the circle nodes' labels.}
    \label{fig:virtualBinaryTree6}
\end{subfigure}%
\caption{For the example, consider $I = 6$ and two nodes $u,v$ with IDs 3 and 5 respectively. Node $u$ is awake in rounds 3, 4 and 5 and $v$ is awake in rounds 5 and 6 (but not in round 7, since there are only $I$ rounds). It can be seen that $u$ communicates whether it has joined the MIS to $v$ in round 5.}
\label{fig:figureForVTMIS}
\end{figure}
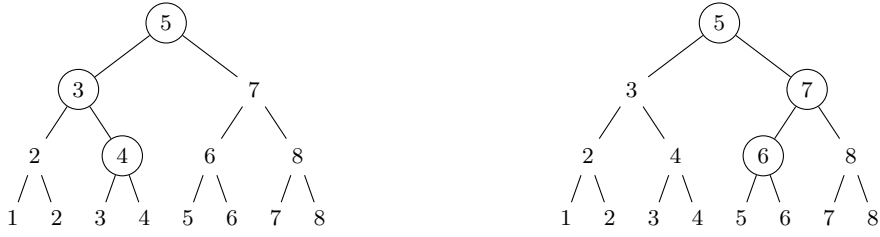

\noindent {\bf \AuxMIStwo{}: MIS in $O(\log n')$ awake time.} We assume that each node knows the value $n'$, an upper bound on the number of nodes of the connected subgraph to which the node belongs. We also assume that the (at most) $n'$ nodes are given unique IDs in $[1,I]$, for some known value $I$ (where $I$ may be exponentially larger than $n'$), and that messages can contain up to $O(\log I)$ bits. First, all nodes participate in the construction of an LDT. (The correctness is guaranteed by the fact that all nodes know the same bound $n'$.) Second, all nodes participate in an operation to (i) compute the exact number of nodes $n''$ in the LDT, and (ii) allow each node to know its rank in a given total ordering of the nodes. Third, the root of the LDT locally computes a uniformly random permutation of $[1,n'']$ and broadcasts it in $O((n' \log n')/ \log I)$ consecutive broadcasts over the LDT. (In each broadcast, we can transmit $O(\log I)$ bits of the total $O(n'' \log n'')$ bits. Once all of these bits have been transmitted, ``null'' messages are sent for the remaining broadcasts, if any.) Each node uses its previously computed rank to retrieve its ID in the permutation. Finally, all nodes run \AuxMISone{} using these smaller IDs.

\begin{lemma}
\label{lem:auxiliaryMISLemma}
\AuxMIStwo{} computes an LFMIS with respect to some uniformly random node ordering, and not to the ID-based ordering, with $O(\log n' + (n' \log n')/ \log I)$ awake complexity, $O(n' (\log n') \log^4 I + ((n')^2 \log n')/ \log I)$ round complexity and $O(\log I)$ bit messages. 
\end{lemma}

\begin{proof}
First, note that \AuxMIStwo{} implements (sequential) randomized greedy MIS. Indeed, constructing the LDT, computing $n'$ and sending down smaller IDs is simply equivalent to computing a uniformly random node ordering on the $n'$ nodes. After which, \AuxMISone{} implements sequential greedy MIS with respect to that order by Lemma \ref{lem:awakeDistributedGreedyMIS}.

Now we prove the rest of the lemma statement. Constructing the LDT takes $O(\log n')$ awake complexity, $O(n' (\log n') \log^4 I)$ round complexity and $O(\log I)$ bit messages (by Lemma \ref{lem:LDT-slow-round-construction}). The ranking operation takes $O(1)$ awake complexity, $O(n')$ round complexity, and $O(\log I)$ bit messages (by Lemma~\ref{lem:LDT-operation}).  The $O((n' \log n')/ \log I)$ consecutive broadcasts used to transmit the permutation chosen by the root (at most $O(n' \log n')$ bits), via messages of $O(\log I)$ bits, take altogether $O((n' \log n')/ \log I)$ awake complexity and $O(((n')^2 \log n')/ \log I)$ round complexity (by Lemma~\ref{lem:LDT-operation}).
Finally, computing the LFMIS (via \AuxMISone{}) with respect to the new node ordering (from the new IDs in $[1,n']$) takes $O(\log n')$ awake complexity, $O(n')$ round complexity, and $O(\log n') = O(\log I)$ bit messages. 
\end{proof}

By replacing the awake efficient LDT construction (Algorithm \LDTCONSTRUCTAWAKE{}) with the round efficient one (Algorithm \LDTCONSTRUCTROUND{}, whose properties are described in Lemma \ref{lem:ldt-fast-round-construction}), we obtain a round efficient version of \AuxMIStwo{}, which we call \AuxMIStwoROUND{}. 
Its properties --- a faster round complexity but larger awake complexity than \AuxMIStwo{} --- are formally stated in the following corollary.

\begin{corollary}
\label{cor:auxiliaryMISLemma}
\AuxMIStwoROUND{} computes an LFMIS with respect to some uniformly random node ordering with $O((\log n') \log^* I + (n' \log n')/ \log I)$ awake complexity, $O((n' \log n') \log^* I + ((n')^2 \log n')/ \log I)$ round complexity and $O(\log I)$ bit messages.
\end{corollary} 

\section{Randomized Greedy MIS in \texorpdfstring{$O(\log \log n)$}{O(log log n)} Awake Rounds}\label{sec:MIS}

We present our main result: an $O(\log \log n)$ awake complexity randomized MIS algorithm. 
We first give a high-level description. Algorithm \AwakeMIS{} computes the lexicographically first MIS, with respect to some uniformly random ordering, in ``batches''. More precisely, the LFMIS is computed over the first $t$ nodes, then over the next $t'$ nodes, and so on. To (energy efficiently) ensure all batches know which nodes (of prior batches) are in the MIS, we coordinate communication using the virtual binary tree technique with only $O(\log \log n)$ awake complexity. Moreover, by batching nodes we can leverage the following ``graph shattering'' property: the subgraph induced by the yet undecided nodes within each batch can be decomposed into small $O(\log n)$-sized connected components. Finally, for each such component, it suffices to run \AuxMIStwo{} from Section \ref{sec:prelim} to compute the LFMIS with respect to a uniformly random ordering (of the component's nodes) in $O(\log \log n)$ awake time. 
\smallskip

\noindent {\bf Description of \AwakeMIS{}.} 
Let $\ell = \lceil \log n - \log \log n \rceil$ and $\Delta' = O(\log n)$ be some parameters decided in the analysis. The output variable $state$ takes value in $\{undecided, inMIS, notinMIS\}$. The MIS problem is said to be solved if all nodes have chosen a state in $\{inMIS, notinMIS\}$ and  the set of all nodes with the $inMIS$ output forms an MIS. 

Initially, each node starts in the ``undecided'' state. Moreover, each node $v \in V$ picks a pair $p(v) = (i,j) \in [1,\ell] \times [1,2\Delta']$ at random (but not uniformly) which decides what batch the node falls in---that batch is denoted by $B_{i,j}$. Batches are ordered via lexicographical order (of the pairs). Next, we describe precisely how nodes choose a batch. Each node chooses $i \in [1,\ell-1]$ with probability $(10 \cdot 2^i \log n) / n$ and $i = \ell$ with the remaining probability. Moreover, each node chooses $j \in [1,2\Delta']$ uniformly at random, that is, with probability $1/(2 \Delta')$.

After the batches have been decided, there are $2 \ell \Delta' = O(\log^2 n)$ phases. (The number of rounds per phase is determined in the analysis, allowing some nodes to sleep through the phase.) In each phase $(i,j) \in [1,\ell] \times [1,2\Delta']$, the first \emph{communication round} is used to update which of the batch's nodes have a neighbor in the MIS. (Let $g: [1,\ell] \times [1,2\Delta'] \mapsto [1,2 \ell \Delta']$ be the natural bijection that preserves lexicographic order.) In more detail, node $v \in V$ is awake if $g(i,j) \in S_{g(p(v))}([1,2 \ell \Delta'])$, and is asleep otherwise---see Subsection \ref{subsec:virtualBinaryTree} for the formal definition of the communication set $S_{g(p(v))}([1,2 \ell \Delta'])$ and Subsection \ref{subsec:auxMIS} for an example. Awake nodes send their state to their neighbors. At the end of the round, awake nodes that received a $inMIS$ message from a neighboring node set their state to $notInMIS$ (thus becoming decided).
After which, the remaining rounds are used by all undecided batch nodes (i.e., with $p(v) = (i,j)$ and in the ``undecided'' state) to execute \AuxMIStwo{} described in Subsection \ref{subsec:auxMIS}. (In the analysis, we show that w.h.p., the remaining rounds are sufficient for this algorithm to terminate.)

\begin{algorithm}[th]
\footnotesize
\caption{{\AwakeMIS{}} for node $v$}
\label{alg:AwakeMIS} 
\begin{algorithmic}[1]
\State \textbf{Input:} $n$
\State $v$ chooses $i \in [1,\ell]$ with probability $(10 \cdot 2^i \log n) / n$ and $j \in [1,2 \Delta']$ with probability $1/(2 \Delta')$ 
\State $p(v) := (i,j)$, $state_v := undecided$

\For{phase $p = 1$ to $2 \ell \Delta'$} 
    \Statex \Comment{For the first (communication) round:}
    \If{$p \in S_{g(p(v))}([1,2 \ell \Delta'])$} \hfill  \Comment{Nodes with $g(p(v)) \neq p$ can also participate}
        \If{$state_v = undecided$}
            \State Listen for one round
            \If{$v$ receives an ``in MIS'' message}
                ${}$\hspace{0.6em} $state_v := not in MIS$
            \EndIf
        \Else
            ${}$\hspace{0.9em} Send $state_v$ to all neighbors
        \EndIf
    \Else
        ${}$\hspace{0.9em} Sleep for one round
    \EndIf

    \Statex \Comment{For the next $O(\log^5 n \log \log n)$ rounds:}
    \If{$p \neq g(p(v))$ or $state_v \neq undecided$}
        ${}$\hspace{0.6em} Sleep for the remainder of the phase
    \Else
        ${}$\hspace{0.9em} $state_v = \AuxMIStwo{}()$
    \EndIf
\EndFor

\end{algorithmic}
\end{algorithm}

\noindent {\bf Analysis of \AwakeMIS{}.} Next, we show correctness and complexity bounds of \AwakeMIS{}. Note that we assume nodes know $n$ exactly (or at least some constant factor approximation) in the above description and below analysis for clarity. This assumption can be removed through straightforward changes, so that nodes can use instead a polynomial upper bound $N$ on $n$.

\begin{theorem}
\label{thm:MISTheorem}
MIS can be solved (w.h.p.) in $O(\log \log n)$ awake complexity and $O(\log^7 n \log \log n)$ round complexity in $\CONGEST$.
\end{theorem}

\begin{proof}
Let us start with some notations. For any $(i,j) \in [1,\ell] \times [1,2\Delta']$, denote by $V_{i,j}$ the union of all batches up to, and including, batch $(i,j)$, and for any $i \in [1,\ell]$, let $V_i = V_{i,2\Delta'}$. Then, we bound the size of $V_i$ for any $i \in [1,\ell]$. 
Recall that each node is in $V_i$ independently and with probability $\sum_{k=1}^i (10 \cdot 2^k \log n)/n = (10 (2^{i+1}-1) \log n)/n$. Hence, by linearity of expectation, $\mathbb{E}[|V_i|] = 10 (2^{i+1}-1) \log n$. After which, we apply the Chernoff bound (for the two tails, see Lemma \ref{lem:chernoffBound}) with $\delta = 1/2$: $\Pr[|V_i| \geq 15 (2^{i+1}-1) \log n] \leq \exp( - (10 (2^{i+1}-1) \log n)/10)$ and $\Pr[|V_i| \leq 5 (2^{i+1}-1) \log n] \leq \exp(- (10 (2^{i+1}-1) \log n)/8 )$. Thus, by union bound, $5 (2^{i+1}-1) \log n \leq |V_i| \leq 15 (2^{i+1}-1) \log n$ holds with probability at least $1-1/n^3$.

We show by induction that by the end of phase $(i,j) \in [1,\ell] \times [1,2\Delta']$, the algorithm has computed the LFMIS over $G[V_{i,j}]$ with respect to a uniformly random ordering of $V_{i,j}$ and with probability at least $1-1/n$. 
Consider the base case. During the first phase, different connected components $C_1,\ldots,C_k$ of $G[V_{1,1}]$ run independent \AuxMIStwo{} executions. Given the upper bound on $V_1$ shown above, these components are of size $O(\log n)$ with probability at least $1-1/n^3$. By setting the number of rounds within the first phase accordingly --- to $O(\log^5 n \log \log n)$ rounds --- all \AuxMIStwo{} executions terminate. Thus, by Lemma \ref{lem:auxiliaryMISLemma}, the output of \AuxMIStwo{} for each component $C_h$ is an LFMIS with respect to a uniformly random node ordering (of $C_h$), which we denote $M_h$. Note that $M = \bigcup_{h=1}^k M_h$ is exactly the set of nodes with state $inMIS$ when the first phase ends. Clearly, $M$ is an MIS since two nodes from different connected components are non-adjacent. 
It remains to show that $M$ is also a LFMIS with respect to a uniformly random ordering of $V_{1,1}$. Consider the following node ordering. First, pick some node $w$ uniformly at random among all $|V_{1,1}|$ nodes and then choose the first node of the ordering uniformly at random among all nodes of the component $w$ belongs to. Next, pick some node $w'$ uniformly at random among the remaining $|V_{1,1}|-1$ nodes and then choose the second node of the ordering uniformly at random among all nodes in the component $w'$ belongs to (excluding the first node in the ordering). Repeat this process until all nodes in $V_{1,1}$ have been ordered. It is straightforward to show that this node ordering is a uniformly random ordering of $V_{1,1}$. Moreover, two nodes $u,v$ from different components are not connected in $G[V_{1,1}]$ and thus whether $u$ is ordered before $v$ (or inversely) does not influence the resulting LFMIS over $G[V_{1,1}]$. Hence, $M$ is the LFMIS over $G[V_{1,1}]$ with respect to that uniformly random ordering of $V_{1,1}$, concluding the base case. 

Now, consider some phase $(i,j) \neq (1,1)$ and assume that the induction hypothesis holds for the (lexicographically) previous phase $(i',j')$. By the induction hypothesis, the algorithm has computed the LFMIS $M'$ over $G[V_{i',j'}]$ with respect to some uniformly random order of $V_{i',j'}$ by the end of phase $(i',j')$.
Since $V_{i',j'} = V_{i,j} \setminus B_{i,j}$, the induction step follows if we show that the algorithm computes a LFMIS over $G[B_{i,j}^*]$ with respect to a uniformly random ordering of $B_{i,j}^*$, where $B_{i,j}^*$ are the nodes of $B_{i,j}$ with no neighbors in $M'$.
To do so, we first note that nodes know whether they are in $B_{i,j}^*$ or not by the end of the communication round of phase $(i,j)$. This can be shown using properties of the communication sets (as in the proof of Lemma \ref{lem:awakeDistributedGreedyMIS}).
In the remainder of the phase, different connected components $C_1,\ldots,C_{k'}$ of $G[B_{i,j}^*]$ run independent \AuxMIStwo{} executions. If $i = 1$, then just as in the base case, these connected components have size $O(\log n)$ with probability at least $1-1/n^3$, and thus $O(\log^5 n \log \log n)$ rounds are sufficient for the \AuxMIStwo{} executions to terminate with probability at least $1-1/n^3$. Otherwise, if $i > 1$, let $M_{i-1}$ be the MIS computed by the end of phase $(i-1,2\Delta')$. Due to the induction hypothesis, we can apply Lemma \ref{lem:sparsityLemmaForRandomOrder} and thus $G[V_i \setminus N(M_{i-1})]$ has maximum degree upper bounded by $\frac{|V_i|}{|V_{i-1}|}  \ln (n^4)$ with probability at least $1 - 1/n^3$. Using the above bounds on $V_i$ and $V_{i-1}$, this is at most $9 \ln (n^4)$. By choosing $\Delta'= 9 \ln (n^4)$ (and using the principle of deferred decisions), Lemma \ref{lem:branchingProcess} implies that $G[B_{i_j} \setminus N(M_{i-1})]$ consists of small $O(\log n)$-sized connected components with probability at least $1-1/n^3$. Given that $B_{i,j}^* \subseteq B_{i_j} \setminus N(M_{i-1})$, $G[B_{i_j}^*]$ also consists of small $O(\log n)$-sized connected components with probability at least $1-1/n^3$. And thus, $O(\log^5 n \log \log n)$ rounds are sufficient for the \AuxMIStwo{} executions to terminate with probability at least $1-1/n^3$. 
Next, we point out that one can show, using an argument similar to that of the base case, that the union $M''$ of the computed MIS---which is exactly the set of nodes whose state becomes $inMIS$ during phase $(i,j)$---is a LFMIS with respect to a uniformly random order of $B_{i,j}^*$. Moreover, it is straightforward to extend this ordering of $B_{i,j}^*$ to a uniformly random order of $V_{i,j}$, as long as all nodes in $B_{i,j}$ are ordered after those in $V_{i',j'}$. (In which case, no matter how some node $z \in B_{i,j} \setminus B_{i,j}^*$ is ordered, $z$ is not in the LFMIS.) Consequently, when phase $(i,j)$ ends, the set $M' \cup M''$ is a LFMIS over $G[V_{i,j}]$ with respect to some uniformly random order of $V_{i,j}$ with probability at least $1-1/n^2$ and the induction step follows. (Note that the error probability of all induction steps added together is at most $1/n$.)

To conclude, we note that all communication is done through $O(\log n)$ bit messages, and in particular the communication within the \AuxMIStwo{} calls. Next, we upper bound the awake complexity. Each node $v \in V$ is awake for at most $O(\log \log n)$ communication rounds over all $O(\log^2 n)$ phases. Moreover, within $v$'s chosen phase $p(v)$, recall that the subgraph induced by the awake nodes is composed of $O(\log n)$-sized connected components with high probability. Then, by Lemma \ref{lem:auxiliaryMISLemma}, $v$ is awake for at most $O(\log \log n + ((\log n) \log \log n)/ \log n) = O(\log \log n)$ rounds (w.h.p.) during the \AuxMIStwo{} execution. 
Finally, a simple modification of \AuxMIStwo{} limiting the number of rounds a node can be awake to $O(\log \log n)$ (where the precise value can be obtained from the analysis and the desired error probability) implies any failure affects correctness rather than the awake complexity. Therefore, the awake complexity of \AwakeMIS{} is (deterministically) upper bounded by $O(\log \log n)$. 
As for the round complexity, the above analysis implies that each phase takes $O(\log^5 n \log \log n)$ rounds. Hence, the algorithm's round complexity is $O(\log^7 n \log \log n)$.
\end{proof}

By using the more round efficient \AuxMIStwoROUND{} (see Corollary \ref{cor:auxiliaryMISLemma}) instead of \AuxMIStwo{}, an MIS can be computed with a  better round complexity, but at the cost of a small $O(\log^* n)$ overhead to the awake complexity. 

\begin{corollary}
MIS can be solved (w.h.p.) in $O((\log \log n) \log^* n)$ awake complexity and \\ $O((\log^3 n) (\log \log n) \log^* n)$ round complexity in $\CONGEST$.
\end{corollary}

\begin{proof}
When using \AuxMIStwoROUND{} instead of \AuxMIStwo{}, phases of $O((\log n) (\log \log n) \log^* n)$ rounds suffice (see Corollary \ref{cor:auxiliaryMISLemma}).
\end{proof}

\section{Conclusion}\label{sec:conclusion}
In this paper, we  show that the fundamental MIS problem on general graphs can be solved in $O(\log \log n)$ awake complexity, i.e., the worst-case number of awake (non-sleeping) rounds taken by all nodes is $O(\log \log n)$. This is the first such result that we are aware of where we can obtain even a $o(\log n)$ bound on the awake  complexity for MIS. A long-standing open question is whether a similar bound (i.e., $o(\log n)$) can be shown for the round complexity.

Several open problems arise from our work. An important one is determining whether one can improve the awake complexity bound of $O(\log \log n)$, or showing that is optimal by showing a lower bound. Another one is whether one can obtain an $O(\log \log n)$ awake complexity MIS algorithm that has $O(\log n)$ round complexity. More generally, can one obtain good trade-offs between awake and round complexity of MIS?

Finally, it would be useful to design algorithms for other symmetry breaking problems such as maximal matching, coloring, etc., that have better awake complexity compared to the traditional round complexity.

\bibliographystyle{plainurl}
\bibliography{references,reference,biblio}

\begin{thebibliography}{10}

\bibitem{Alon_1986}
Noga Alon, L{\'{a}}szl{\'{o}} Babai, and Alon Itai.
\newblock A fast and simple randomized parallel algorithm for the maximal independent set problem.
\newblock {\em Journal of Algorithms}, 7(4):567--583, 1986.

\bibitem{AMP22}
John Augustine, William~K. {Moses Jr.}, and Gopal Pandurangan.
\newblock Brief announcement: Distributed {MST} computation in the sleeping model: Awake-optimal algorithms and lower bounds.
\newblock {\em Proceedings of the 41st ACM Symposium on Principles of Distributed Computing (PODC)}, pages 51--53, 2022.
\newblock Full version available on arXiv: https://arxiv.org/abs/2204.08385.

\bibitem{Balliu_2019}
Alkida Balliu, Sebastian Brandt, Juho Hirvonen, Dennis Olivetti, Mika\"{e}l Rabie, and Jukka Suomela.
\newblock Lower bounds for maximal matchings and maximal independent sets.
\newblock In {\em IEEE FOCS}, pages 481--497, November 2019.

\bibitem{MIS-trees-lower}
Alkida Balliu, Sebastian Brandt, Fabian Kuhn, and Dennis Olivetti.
\newblock Improved distributed lower bounds for {MIS} and bounded (out-)degree dominating sets in trees.
\newblock In {\em {PODC} '21: {ACM} Symposium on Principles of Distributed Computing}, pages 283--293. {ACM}, 2021.

\bibitem{barenboim}
Leonid Barenboim and Michael Elkin.
\newblock Sublogarithmic distributed {MIS} algorithm for sparse graphs using {Nash-Williams} decomposition.
\newblock {\em Distributed Comput.}, 22(5-6):363--379, 2010.

\bibitem{Barenboim_2016}
Leonid Barenboim, Michael Elkin, Seth Pettie, and Johannes Schneider.
\newblock The locality of distributed symmetry breaking.
\newblock {\em Journal of the {ACM}}, 63(3), June 2016.
\newblock Conference version: {\em IEEE FOCS} 2012.

\bibitem{BM21}
Leonid Barenboim and Tzalik Maimon.
\newblock Deterministic logarithmic completeness in the distributed sleeping model.
\newblock In {\em 35th International Symposium on Distributed Computing, {DISC}}, volume 209, pages 10:1--10:19, 2021.

\bibitem{young3}
Michael~A. Bender, Jeremy~T. Fineman, Mahnush Movahedi, Jared Saia, Varsha Dani, Seth Gilbert, Seth Pettie, and Maxwell Young.
\newblock Resource-competitive algorithms.
\newblock {\em {SIGACT} News}, 46(3):57--71, 2015.

\bibitem{Blelloch_2012}
Guy~E. Blelloch, Jeremy~T. Fineman, and Julian Shun.
\newblock Greedy sequential maximal independent set and matching are parallel on average.
\newblock In {\em ACM Symposiumon Parallelism in Algorithms and Architectures (SPAA)}, pages 308--317, 2012.

\bibitem{Chang23}
Yi-Jun Chang.
\newblock The energy complexity of diameter and minimum cut computation in bounded-genus networks.
\newblock In Sergio Rajsbaum, Alkida Balliu, Joshua~J. Daymude, and Dennis Olivetti, editors, {\em Structural Information and Communication Complexity (SIROCCO)}, pages 262--296, Cham, 2023. Springer Nature Switzerland.

\bibitem{energy2}
Yi{-}Jun Chang, Varsha Dani, Thomas~P. Hayes, Qizheng He, Wenzheng Li, and Seth Pettie.
\newblock The energy complexity of broadcast.
\newblock In {\em ACM PODC}, pages 95--104, 2018.

\bibitem{energy3}
Yi{-}Jun Chang, Varsha Dani, Thomas~P. Hayes, and Seth Pettie.
\newblock The energy complexity of {BFS} in radio networks.
\newblock In {\em ACM PODC}, pages 273--282, 2020.

\bibitem{CDJ23}
Yi-Jun Chang, Ran Duan, and Shunhua Jiang.
\newblock Near-optimal time–energy tradeoffs for deterministic leader election.
\newblock {\em ACM Trans. Algorithms}, 19(4), sep 2023.
\newblock \href {https://doi.org/10.1145/3614429} {\path{doi:10.1145/3614429}}.

\bibitem{CS22}
Yi-Jun Chang and Shunhua Jiang.
\newblock The energy complexity of las vegas leader election.
\newblock In {\em Proceedings of the 34th ACM Symposium on Parallelism in Algorithms and Architectures}, SPAA '22, page 75–86, New York, NY, USA, 2022. Association for Computing Machinery.
\newblock \href {https://doi.org/10.1145/3490148.3538586} {\path{doi:10.1145/3490148.3538586}}.

\bibitem{energy1}
Yi{-}Jun Chang, Tsvi Kopelowitz, Seth Pettie, Ruosong Wang, and Wei Zhan.
\newblock Exponential separations in the energy complexity of leader election.
\newblock {\em {ACM} Trans. Algorithms}, 15(4):49:1--49:31, 2019.
\newblock Conference version: {\em ACM STOC} 2017.

\bibitem{podc2020}
Soumyottam Chatterjee, Robert Gmyr, and Gopal Pandurangan.
\newblock Sleeping is efficient: {MIS} in \emph{O}(1)-rounds node-averaged awake complexity.
\newblock In {\em {ACM} Symposium on Principles of Distributed Computing, PODC}, pages 99--108, 2020.

\bibitem{Coppersmith_1989}
Don Coppersmith, Prabhakar Raghavan, and Martin Tompa.
\newblock Parallel graph algorithms that are efficient on average.
\newblock {\em Information and Computation}, 81(3):318--333, 1989.

\bibitem{energy4}
Varsha Dani, Aayush Gupta, Thomas~P. Hayes, and Seth Pettie.
\newblock {Wake up and Join Me! an Energy-Efficient Algorithm for Maximal Matching in Radio Networks}.
\newblock In {\em 35th International Symposium on Distributed Computing (DISC)}, pages 19:1--19:14, 2021.

\bibitem{energy5}
Varsha Dani and Thomas~P. Hayes.
\newblock How to wake up your neighbors: Safe and nearly optimal generic energy conservation in radio networks.
\newblock In Christian Scheideler, editor, {\em 36th International Symposium on Distributed Computing, {DISC} 2022, October 25-27, 2022, Augusta, Georgia, {USA}}, volume 246 of {\em LIPIcs}, pages 16:1--16:22, 2022.

\bibitem{DMP22}
Fabien Dufoulon, William~K Moses~Jr., and Gopal Pandurangan.
\newblock Sleeping is superefficient: {MIS} in exponentially better awake complexity.
\newblock {\em arXiv preprint arXiv:2204.08359}, 2022.

\bibitem{Feeney_2001}
Laura~Marie Feeney and Martin Nilsson.
\newblock Investigating the energy consumption of a wireless network interface in an ad hoc networking environment.
\newblock In {\em IEEE INFOCOM}, volume~3, pages 1548--1557, 2001.

\bibitem{Fischer_2018}
Manuela Fischer and Andreas Noever.
\newblock Tight analysis of parallel randomized greedy {MIS}.
\newblock In {\em SODA}, pages 2152--2160, 2018.

\bibitem{FMRT23}
Pierre Fraigniaud, Pedro Montealegre, Ivan Rapaport, and Ioan Todinca.
\newblock Energy-efficient distributed algorithms for synchronous networks.
\newblock In Sergio Rajsbaum, Alkida Balliu, Joshua~J. Daymude, and Dennis Olivetti, editors, {\em Structural Information and Communication Complexity (SIROCCO)}, pages 482--501, Cham, 2023. Springer Nature Switzerland.

\bibitem{Ghaffari_2016_SODA}
Mohsen Ghaffari.
\newblock An improved distributed algorithm for maximal independent set.
\newblock In {\em SODA}, pages 270--277, 2016.

\bibitem{GGR20}
Mohsen Ghaffari, Christoph Grunau, and V{\'{a}}clav Rozhon.
\newblock Improved deterministic network decomposition.
\newblock In {\em Proceedings of the 2021 {ACM-SIAM} Symposium on Discrete Algorithms, {SODA}}, pages 2904--2923, 2021.

\bibitem{ghaffari-sleeping}
Mohsen Ghaffari and Julian Portmann.
\newblock Average awake complexity of {MIS} and matching.
\newblock In {\em ACM Symposium on Parallelism in Algorithms and Architectures (SPAA)}, pages 45--55, 2022.

\bibitem{ghaffari-podc2023}
Mohsen Ghaffari and Julian Portmann.
\newblock Distributed {MIS} with low energy and time complexities.
\newblock In Rotem Oshman, Alexandre Nolin, Magn{\'{u}}s~M. Halld{\'{o}}rsson, and Alkida Balliu, editors, {\em Proceedings of the 2023 {ACM} Symposium on Principles of Distributed Computing, {PODC} 2023, Orlando, FL, USA, June 19-23, 2023}, pages 146--156. {ACM}, 2023.
\newblock \href {https://doi.org/10.1145/3583668.3594587} {\path{doi:10.1145/3583668.3594587}}.

\bibitem{young2}
Seth Gilbert, Valerie King, Seth Pettie, Ely Porat, Jared Saia, and Maxwell Young.
\newblock (near) optimal resource-competitive broadcast with jamming.
\newblock In Guy~E. Blelloch and Peter Sanders, editors, {\em 26th {ACM} Symposium on Parallelism in Algorithms and Architectures, {SPAA} '14, Prague, Czech Republic - June 23 - 25, 2014}, pages 257--266. {ACM}, 2014.

\bibitem{young1}
Seth Gilbert and Maxwell Young.
\newblock Making evildoers pay: resource-competitive broadcast in sensor networks.
\newblock In Darek Kowalski and Alessandro Panconesi, editors, {\em {ACM} Symposium on Principles of Distributed Computing, {PODC} '12, Funchal, Madeira, Portugal, July 16-18, 2012}, pages 145--154. {ACM}, 2012.

\bibitem{hourani-mis}
Khalid Hourani, Gopal Pandurangan, and Peter Robinson.
\newblock Awake-efficient distributed algorithms for maximal independent set.
\newblock In {\em IEEE Conference on Distributed Computing Systems (ICDCS)}, pages 1338--1339, 2022.

\bibitem{jurdin}
Tomasz Jurdzinski, Miroslaw Kutylowski, and Jan Zatopianski.
\newblock Efficient algorithms for leader election in radio networks.
\newblock In Aleta Ricciardi, editor, {\em Proceedings of the Twenty-First Annual {ACM} Symposium on Principles of Distributed Computing, {PODC} 2002, Monterey, California, USA, July 21-24, 2002}, pages 51--57. {ACM}, 2002.

\bibitem{pajak}
Marcin Kardas, Marek Klonowski, and Dominik Pajak.
\newblock Energy-efficient leader election protocols for single-hop radio networks.
\newblock In {\em 42nd International Conference on Parallel Processing, {ICPP} 2013, Lyon, France, October 1-4, 2013}, pages 399--408. {IEEE} Computer Society, 2013.

\bibitem{King_2011}
Valerie King, Cynthia~A. Phillips, Jared Saia, and Maxwell Young.
\newblock Sleeping on the job: Energy-efficient and robust broadcast for radio networks.
\newblock {\em Algorithmica}, 61(3):518--554, 2011.

\bibitem{K18}
Christian Konrad.
\newblock {MIS} in the congested clique model in {$O(\log \log \Delta)$} rounds.
\newblock {\em arXiv preprint arXiv:1802.07647}, 2018.

\bibitem{tcsgnp}
K.~Krzywdzi{\'{n}}ski and K.~Rybarczyk.
\newblock Distributed algorithms for random graphs.
\newblock {\em Theoretical Computer Science}, 605:95--105, 2015.

\bibitem{Kuhn_2016}
Fabian Kuhn, Thomas Moscibroda, and Roger Wattenhofer.
\newblock Local computation: Lower and upper bounds.
\newblock {\em Journal of the {ACM}}, 63(2), 2016.

\bibitem{lenzen}
Christoph Lenzen and Roger Wattenhofer.
\newblock {MIS} on trees.
\newblock In {\em ACM PODC}, pages 41--48, 2011.

\bibitem{Luby_1986}
Michael Luby.
\newblock A simple parallel algorithm for the maximal independent set problem.
\newblock {\em {SIAM} Journal on Computing}, 15(4):1036--1053, 1986.
\newblock Conference version: {\em ACM STOC 1985}.

\bibitem{maimon}
Tzalik Maimon.
\newblock Sleeping model: Local and dynamic algorithms, 2021.
\newblock \href {https://arxiv.org/abs/2112.05344} {\path{arXiv:2112.05344}}.

\bibitem{Upfalbook}
Michael Mitzenmacher and Eli Upfal.
\newblock {\em Probability and computing: randomization and probabilistic techniques in algorithms and data analysis}.
\newblock Cambridge university press, 2017.

\bibitem{Murthy_2004_Book}
Chebiyyam Siva~Ram Murthy and Balakrishnan Manoj.
\newblock {\em Ad Hoc Wireless Networks: Architectures and Protocols}.
\newblock Prentice Hall PTR, USA, 2004.

\bibitem{nakano}
Koji Nakano and Stephan Olariu.
\newblock Randomized leader election protocols in radio networks with no collision detection.
\newblock In D.~T. Lee and Shang{-}Hua Teng, editors, {\em Algorithms and Computation, 11th International Conference, {ISAAC} 2000, Taipei, Taiwan, December 18-20, 2000, Proceedings}, volume 1969 of {\em Lecture Notes in Computer Science}, pages 362--373. Springer, 2000.

\bibitem{Peleg_2000_Book}
David Peleg.
\newblock {\em Distributed Computing: A Locality-sensitive Approach}.
\newblock Society for Industrial and Applied Mathematics, 2000.

\bibitem{Rozhon2020}
V{\'{a}}clav Rozhon and Mohsen Ghaffari.
\newblock Polylogarithmic-time deterministic network decomposition and distributed derandomization.
\newblock In {\em 52nd Annual {ACM} {SIGACT} Symposium on Theory of Computing, {STOC}}, pages 350--363, 2020.

\bibitem{Wang_2006}
Qin Wang, Mark Hempstead, and Woodward Yang.
\newblock A realistic power consumption model for wireless sensor network devices.
\newblock In {\em The Third Annual {IEEE} Communications Society on Sensor and Ad Hoc Communications and Networks}, volume~1 of {\em SECON '06}, pages 286--295, September 2006.

\bibitem{Yang_2013}
Ou~Yang and Wendi Heinzelman.
\newblock An adaptive sensor sleeping solution based on sleeping multipath routing and duty-cycled mac protocols.
\newblock {\em {ACM} Transactions on Sensor Networks}, 10(1), 2013.

\bibitem{Zheng_2005}
Rong Zheng and Robin Kravets.
\newblock On-demand power management for ad hoc networks.
\newblock {\em Ad Hoc Networks}, 3(1):51--68, 2005.

\end{thebibliography}


\newpage
\appendix
\section{Labeled Distance Tree (LDT) and a Faster Round  Complexity}
\label{sec:app-ldt}
In this section, we describe a useful data structure called a labeled distance tree (LDT), introduced in~\cite{AMP22}. We first describe the structure and relevant procedures needed to construct it in Section~\ref{app-subsec:ldt-explain-and-procedures}. Subsequently, we give an informal description on how one might construct an LDT in Section~\ref{app-subsec:construct-ldt}, i.e., we describe algorithm $\LDTCONSTRUCTROUND$. Finally, in Section~\ref{app-subsec:ldt-procedures} we give a few useful procedures that can be run once an LDT is constructed.

\subsection{LDT Description and Relevant Procedures}
\label{app-subsec:ldt-explain-and-procedures}
For a given graph, a \textit{labeled distance tree (LDT)} is a spanning tree of that graph such that (i) all nodes in the tree know the ID of the root of the tree, called the ID of the LDT, (ii) each node knows its depth in the tree (i.e., the hop-distance from itself to the root of the tree via tree edges), and (iii) each node knows the IDs of its parent and children, if any, in the tree. If a given graph can be partitioned into a disjoint set of such LDTs, we refer to that graph as a \textit{forest of labeled distance trees (FLDT)}.

In order to construct an LDT over a given graph, we start from a situation where all nodes are considered LDTs of their own (i.e., the overall graph is an FLDT) and we successively merge LDTs together in phases until we arrive at a situation where all nodes in the graph belong to the same LDT. In the course of this process, we utilize several simple procedures (e.g., upcast, broadcast)
that are modified to be efficient in awake complexity. We first describe those procedures before explaining how to construct an LDT in the next section.

For the procedures described below, it is assumed that the initial graph has already been divided into an FLDT where each node $u$ knows the ID of the root, $root$, of the LDT it belongs to, $u$'s distance to $root$ within the LDT, as well as $u$'s parent and children, if any, in the LDT it belongs to. First of all, we define a transmission schedule that is used in each of the procedures and will be directly utilized in the algorithms. Then we describe the procedures themselves.

\begin{sloppypar}
~\\ \noindent \textbf{Transmission schedule of nodes in an LDT.}
Consider an LDT rooted at the node $root$ and a node $u$ in that tree {\em at distance $i$ from the root}. Let $n$ be an upper bound on the number of nodes in the LDT. We describe a transmission schedule   and an upper bound on the number of nodes in the LDT $n$. The transmission schedule $\TS(root, u, n)$ assigns a set of rounds to $u$ to be awake in from a block of $2n+1$ possible rounds. For ease of explanation, we assign names to each of these rounds as well. For all non-root nodes $u$, the set of rounds that $\TS(root,u)$ assigns to $u$ includes rounds $i,i+1, n+1, 2n-i+1$, and $2n-i+2$ with corresponding names $\DOWNRECEIVE, \DOWNSEND, \SIDESENDRECEIVE, \UPRECEIVE$, and $\UPSEND$, respectively. $\TS(root,root,n)$ assigns to $root$ the set containing only the rounds $1$, $n+1$, and $2n+1$ with names $\DOWNSEND$, $\SIDESENDRECEIVE$, and $\UPRECEIVE$, respectively.\footnote{We assumed that $\TS(\cdot, \cdot, n)$ was started in round $1$ when assigning rounds. If $\TS(\cdot, \cdot, n)$ is started in round $r$, then the correct round numbers can be obtained by adding $r-1$ to the values mentioned in the description.} 
\end{sloppypar}

This transmission schedule can be used to modify typical procedures on a tree (e.g., upcast, broadcast) to procedures that have small awake complexity. During one instance of $\TS(root, u, n)$, by having each node wake up in a carefully selected non-empty subset of its at most $5$ named rounds, we guarantee that all nodes in the LDT have woken up at least once and that information is propagated in the correct ``direction'' in the LDT as needed. We name useful procedures to construct an LDT and give guarantees on these procedures below.

~\\ \noindent \textbf{Broadcasting a message in an LDT.}
Procedure $\DOWNCAST(n)$, run by all nodes in a given LDT, allows the root node of that LDT to transmit a given message to all nodes in the LDT in $O(1)$ awake complexity and $O(n)$ round complexity.

~\\ \noindent \textbf{Upcasting the minimum value in an LDT.}
Procedure $\UPCASTMIN(n)$, run by all nodes in a given LDT, allows the smallest value among all values held by the nodes of that LDT, to be propagated to the root of the tree in $O(1)$ awake complexity and $O(n)$ round complexity. 

~\\ \noindent \textbf{Transmitting a message between nodes of adjacent LDTs.}
Procedure $\TRANSMITADJACENT(n)$, run by all nodes in an LDT, allows each node in the LDT to transfer a message, if any, to neighboring nodes belonging to other LDTs in $O(1)$ awake complexity and $O(n)$ round complexity.

\subsection{Construction of an LDT}
\label{app-subsec:construct-ldt}
\begin{sloppypar}
In this section, we describe how to deterministically construct an LDT in algorithm $\LDTCONSTRUCTROUND$. We note that the original process~\cite{AMP22} was designed to construct a minimum spanning tree while simultaneously constructing an LDT over the original graph. We describe a simplified version of the process that focuses on just constructing an LDT over the original graph.
\end{sloppypar}

~\\ \noindent \textbf{Algorithm.}
At a high level, we first spend a single round so that each node is aware of the IDs of each of its neighbors and each edge can be assigned a unique ID known to both endpoints (the ID can be the composition of the edge's endpoints' IDs in increasing order). Subsequently, we run a version of the classical GHS algorithm, modified to ensure that the awake complexity of nodes is small and to also ensure that at the beginning of each phase of the algorithm, an FLDT is maintained. Initially, each node is considered to be its own LDT (and so the overall graph is an FLDT) and we describe how to merge LDTs together in $O(\log n)$ phases until all nodes in the graph belong to the same LDT.

In each phase, we perform the following sequence of steps in three stages. Recall that at the beginning of the phase, we have a forest of LDTs. 
We explain each phase of the algorithm from a bird's eye perspective with the details of what each LDT does as bullet points. 
\begin{enumerate}
    \item \textbf{Stage 1.} Each LDT finds its ``minimum'' outgoing edge and both nodes of an outgoing edge are made aware of it. (Since the edges are unweighted, this is just the outgoing edge whose ID is smallest among all outgoing edges.) Also, all nodes of an a given LDT should know both the name of the outgoing edge and the ID of the LDT that edge leads to.
    \begin{enumerate}
        \item (The nodes in) each LDT run procedure $\UPCASTMIN(n)$ to collect the ID of the ``minimum'' outgoing edge at the root of the LDT.
        \item Each LDT runs procedure $\DOWNCAST(n)$ to let all nodes in the LDT know the ID of the chosen outgoing edge.
        \item Each LDT $L$ runs $\TRANSMITADJACENT(n)$ to let nodes from other LDTs know if they are part of the chosen outgoing edge from $L$.
    \end{enumerate}
    
    \item \textbf{Stage 2.} Consider the supergraph $H$ where each LDT is a node and the outgoing edges identified in the previous step are the edges. Denote the connected components in $H$ as $C_1$, $C_2$, $\ldots$, $C_p$. Our final goal is to decompose the connected components in $H$ in to a forest $F$ of small-depth trees $SDT_1$, $SDT_2$, $\ldots$, $SDT_k$.\footnote{Note that each tree $SDT_i$ is itself a supergraph consisting of LDTs as nodes and edges between LDTs forming the edges in the supergraph. That is, both $F$ and $H$ are supergraphs with the same set of nodes but possibly different edges. Also note that the number of small-depth trees in $F$, $k$, may be different than the number of connected components in $H$, $p$.} This is done as follows. For each connected component $C_i$ in $H$, the LDTs that make up $C_i$ work together to construct a spanning tree $T_i$ on top of $C_i$. Then the LDTs in each $T_i$ simulate a Cole-Vishkin style coloring of $T_i$ and use this coloring to then perform a maximal matching. Each edge of this maximal matching is added to $F$. Finally, after the maximal matching, if there exist isolated LDTs (i.e., unmatched LDTs), those LDTs add their outgoing edge to their parent LDT to $F$. If the isolated LDT is a root of some $T_i$, then it will not have a parent. In this case, this isolated LDT chooses an outgoing edge to one of its children in $T_i$.
    \begin{enumerate}
        \item Notice that in each connected component $C_i$ of the supergraph $H$, there will exist \textit{exactly} two LDTs with outgoing edges into each other. The LDT with the smaller ID becomes the root of the tree $T_i$ for that connected component and the other LDT becomes its child. By having all LDTs use procedures $\UPCASTMIN(n)$ and $\DOWNCAST(n)$ a constant number of times, the nodes in both of these LDTs identify that they belong to such LDTs in the connected component and can identify if they are in the LDT that forms the root of the tree in $H$. 
        \item Subsequently, each remaining LDT in $C_i$ identifies its outgoing edge as leading to its parent in the spanning tree $T_i$. This does not require communication between nodes as this is just local computation performed by the roots of the LDTs.
        \item Now, a Cole-Vishkin style coloring algorithm to 6-color the LDTs of each $T_i$ (i.e., each LDT in $T_i$ is colored a single color in $[1,6]$, known to all nodes within that LDT) can be easily simulated. Each round of the $O(\log^* n)$ algorithm consists of some local computation and then having each node communicate with each of its children. One round of local computation can be simulated by the roots of each of the LDTs in one round. One round of communication can be simulated in $O(1)$ awake complexity and $O(n)$ round complexity via a constant number of uses of the procedures $\UPCASTMIN(n)$, $\DOWNCAST(n)$, and $\TRANSMITADJACENT(n)$.
        \item Previously, a spanning tree was constructed over each connected component of the supergraph $H$ and each of the LDTs, constituting nodes in this spanning tree, were colored in via a 6-coloring. Now, in order to obtain a maximal matching on these components, the following process is run for $6$ phases. We maintain the invariant that at the beginning of each phase, each LDT (all nodes belonging to that LDT) knows if it is matched or not and for every inter-LDT edge, both nodes of that edge know the matching status of both LDTs. In phase $i$, each unmatched LDT of color $i$ chooses one of its unmatched children, if any, arbitrarily (say by performing procedure $\UPCASTMIN(n)$ to return the smallest inter-LDT edge among edges leading to unmatched children). Subsequently the LDT informs the child of being matched to it and informs all adjacent LDTs that it is no longer unmatched. Each phase can be simulated through a constant number of uses of $\UPCASTMIN(n)$, $\DOWNCAST(n)$, and $\TRANSMITADJACENT(n)$. 
        \item Each LDT that remains unmatched in $T_i$ (except the root LDT of $T_i$) now informs its parent LDT in $T_i$ that the inter-LDT edge between them also belongs to $F$. This can be done through a constant number of uses of $\UPCASTMIN(n)$, $\DOWNCAST(n)$, and $\TRANSMITADJACENT(n)$. 
        \item Finally, if the parent LDT in $T_i$ is unmatched, it chooses one of its children in $T_i$ and informs that child (really the node within that LDT at the other end of the edge) that the inter-LDT edge between them is also in $F$. This can be done through a constant number of uses of $\UPCASTMIN(n)$, $\DOWNCAST(n)$, and $\TRANSMITADJACENT(n)$. 
    \end{enumerate}
    
    \item \textbf{Stage 3.} At the end of the previous stage, a forest $F$ of small-depth trees $SDT_i$ was formed. (These trees consist of the LDTs and any edges added to $F$. While all nodes within an LDT may not be aware of these edges in $F$, the endpoints (nodes within LDTs) of every such edge are aware of it.) The final part of each phase consists of merging together the LDTs in each small-depth tree $SDT_i$ into one large LDT. Care must be taken to ensure that each node of a resulting merged LDT has the correct ID for that LDT (i.e., the ID of the root of the LDT) and furthermore, each node in the merged LDT maintains the correct distance-from-root value. Additionally, each node may need to be re-oriented, i.e., each node may need to update information about who its parent and children in the LDT are.
    \begin{enumerate}
        \item Notice that every tree $SDT_i$ in $F$ is of diameter at most $4$ (i.e., any two LDTs in the same tree $SDT_i$ in $F$ are at most distance $4$ apart).\footnote{To see why this is true, notice how the edges were added to $F$ to construct each $SDT_i$. First, edges were added from a maximal matching. This created small-depth trees of diameter $1$. Now, for each tree $T_i$, isolated (non-root) LDTs added edges to their parents. This could result in small-depth trees of diameter $3$. Finally, for each tree $T_i$, if the root LDT is isolated, it could add an edge to one of its children. This could increase the diameter of the resulting small-depth tree by $1$.} First, for each $SDT_i$, we choose its constituent LDT with the smallest ID to be the core of the merged LDT, around which we re-orient the nodes of the other LDTs. It is easy to see that, for any tree $SDT_i$ in $F$, the smallest LDT ID in $SDT_i$ can be progagated to every LDT in $SDT_i$ through a constant number of instances of  $\UPCASTMIN(n)$, $\DOWNCAST(n)$, and $\TRANSMITADJACENT(n)$. 
        \item Now, all nodes within a given tree $SDT_i$ in $F$ know the ID that will become the ID of the final merged LDT. Let LDT $L$ have this ID. Consider an LDT $L'$ that is adjacent to the LDT $L$ in $SDT_i$. The nodes in $L'$ must re-align themselves and update their distance-to-root values in the merged LDT consisting of $L$ and $L'$. Suppose the edge $(u,v)$ is the inter-LDT edge between $L$ and $L'$ such that $u \in L$ and $v \in L'$. Now, the nodes in $L'$ update their values by utilizing two instances of a transmission schedule parameterized by $n$. Recall that $v$ is aware of $u$'s distance-to-root and so knows its own distance-to-root. In the first instance of the transmission schedule the nodes in the branch from $v$ to the root of $L'$ update their  distance-to-root values and re-orient themselves using the $\UPSEND$ and $\UPRECEIVE$ rounds. In the second instance, the remaining nodes in $L'$ update their values in the $\DOWNSEND$ and $\DOWNRECEIVE$ rounds and can re-orient themselves. This process allows all LDTs at distance $1$ from LDT $L$ in $SDT_i$ to update their values. By running this process $4$ times, we guarantee that all LDTs up to and including distance $4$ from LDT $L$ in $SDT_i$ can update their values and re-orient themselves.
    \end{enumerate}
\end{enumerate}

~\\ \noindent \textbf{Analysis.} We briefly analyze the algorithm above to give the intuition of correctness and complexities of the construction. We note that the full analysis was already given in~\cite{AMP22}. The correctness comes from the fact that we maintain a forest of LDTs at the beginning of each phase. The construction of the forest $F$ in the final part of each phase and the merging of LDTs in $F$ guarantees that a constant number of the LDTs are merged together in each phase. To see this, notice that, in every phase, every tree $SDT_i$ in $F$ consists of at least two LDTs, resulting in at least a constant fraction of the LDTs ``disappearing'' in the phase (but really just being merged into one another).\footnote{Notice that each tree $SDT_i$ in $F$ is formed from the LDTs in $H$ as well as edges added to $F$ in stage two. In stage two, every LDT in each tree $T_i$ is either unmatched or matched with another LDT. Each LDT that is matched ensures that the small-depth tree its belongs to comprises of at least two LDTs. For each LDT that remained unmatched, it is still guaranteed to merge with another LDT because it will either add an edge to its parent in $T_i$ to $F$ (if it is not the root of $T_i$) or else add an edge to one of its children in $T_i$ to $F$ (if it is the root of $T_i$).} As a result, after $O(\log n)$ phases, all LDTs are merged together into one single LDT.

Regarding the awake complexity and round complexity, we see that the first and third stage each use a constant number of transmission schedules parameterized by $n$, as well as a constant number of calls to $\UPCASTMIN(n)$, $\DOWNCAST(n)$, and $\TRANSMITADJACENT(n)$. Totally, all of these contribute $O(1)$ awake complexity and $O(n)$ round complexity. However, in the second stage, the process of simulating the Cole-Vishkin style coloring takes a total of $O(\log^* n)$ awake complexity and $O(n \log^* n)$ round complexity. Since, there are a total of $O(\log n)$ phases, we see that the total awake complexity is $O((\log n) \log^* n)$ and the total round complexity is $O(n (\log n) \log^* n)$.

The following lemma captures the above guarantees.

\begin{lemma}
For a given graph of at most $n$ nodes, where $n$ is known to all the nodes, $\LDTCONSTRUCTROUND$ deterministically constructs an LDT over the given graph in the $\CONGEST$ setting in $O((\log n) \log^* n)$ awake complexity and $O(n (\log n) \log^* n)$ round complexity.
\end{lemma}

\subsection{Useful Procedures}\label{app-subsec:ldt-procedures}

Once an LDT is constructed, we can then leverage it to run fast awake complexity procedures. In particular, we describe the operations of broadcast and ranking from Definition~\ref{def:ldt-ops}.

~\\ \noindent \textbf{Broadcasting in an LDT.}
We describe Procedure~$\LDTBROADCAST$, run by each node $v$ of an LDT, which takes the message $msg_r$ of the root of the LDT as input and results in all nodes of the LDT receiving $msg_r$. Consider an LDT of size at most $n'$, where the root of the LDT wants to broadcast a message of size at most $m_r$ bits to all nodes in the LDT. Assume that messages of size $O(m_r)$ can be sent over each edge. Broadcast can be performed by having all nodes in the LDT participate in one instance of a transmission schedule parameterized by $n'$, where each node receives this message in their $\DOWNRECEIVE$ round and propagates the message to its children in its $\DOWNSEND$ round. This takes $O(1)$ awake complexity and $O(n')$ round complexity.

~\\ \noindent \textbf{Ranking in an LDT.}
We describe Procedure~$\LDTRANK$, run by each node $v$ of an LDT, which calculates $v$'s rank in some overall total order of the nodes of the LDT, and additionally allows $v$ to learn the size of the LDT. (The given total ordering may be intuitively understood as the in-order ranking of a binary tree, when extended to an $n$-ary tree. For an $n$-ary tree, one possible generalization, that we use here, would be to recursively do the following: visit the leftmost subtree, then the root, then the remaining subtrees in some arbitrary order.) Consider an LDT of size at most $n'$. The procedure consists of two instances of a transmission schedule parameterized by $n'$. In the first instance of a transmission schedule, each node listens in its $\UPRECEIVE$ round for the number of nodes in each of its children's subtrees and remembers these values, and then in its $\UPSEND$ round it sends the sum of these values plus one to its parent. At the end of this instance of a transmission schedule, the root will know the total number of nodes in the tree $n''$. In the second instance of the transmission schedule, each node $v$ receives in its $\DOWNRECEIVE$ round the value of $n''$ and  a value from its parent, which $v$ uses to calculate its rank in a manner described below, and subsequently sends down the value of $n''$ and unique values to each of $v$'s children in its $\DOWNSEND$ round. For a given node $v$ that receives a value $x$ and has $k$ children $c_1, c_2, \ldots, c_k$ with corresponding number of nodes in their subtrees $n_1, n_2, \ldots, n_k$, do the following. We assume that the root of the LDT ``receives'' $x=0$.\footnote{Each node maintains an arbitrary internal ordering of its children in the LDT, if any.} Set $v$'s rank to $x+n_1$. During $v$'s $\DOWNSEND$ round, it sends down $x$ to $c_1$, and for $i \in [2,k]$, $v$ sends the value $1 + \sum_{j=1}^{i-1} n_j$ to $c_i$. This procedure takes $O(1)$ awake complexity and $O(n')$ round complexity.

The guarantees on constructing the LDT as well as the procedures described above are reflected in the following lemma.

\begin{lemma}\label{app:procedures-ldt-guarantees}
For any connected $V' \subseteq V$ of at most $n'$ nodes, where $n'$ is known to all nodes, with unique IDs in $[1,I]$:
\begin{itemize}
    \item An LDT over $G[V']$ can be constructed deterministically with $O(\log I)$ bit messages, with $O((\log n') \log^* I)$ awake complexity and $O(n' (\log n') \log^* I)$ round complexity.
    \item Broadcast over an LDT --- in which the LDT root $v_r$ has a message of size at most $m_r$ --- can be executed deterministically with $O(m_r)$ bit messages, with $O(1)$ awake complexity and $O(n')$ round complexity.
    \item Ranking over an LDT --- in which each node $v \in V'$ learns its rank in some total ordering of the tree and also learns $|V'|$ --- can be executed deterministically with $O(\log I)$ bit messages, with $O(1)$ awake complexity and $O(n')$ round complexity.
\end{itemize} 
\end{lemma}

\end{document}